\documentclass[10pt, onecolumn]{IEEEtran}
\usepackage{amsfonts,amsmath,amssymb,amsthm,caption}
\usepackage{algorithm}
\usepackage{algpseudocode}
\usepackage{amsbsy}
\usepackage{graphicx,multirow,bm}
\usepackage{color}
\makeatletter
\newtheoremstyle{mythm}{3pt}{3pt}{}{16pt}{\bfseries}{:}{.5em}{}
\theoremstyle{mythm}
\newtheorem{theorem}{Theorem}
\newtheorem{example}{Example}

\newtheorem{remark}{Remark}

\newtheorem{proposition}{Proposition}

\newtheorem{lemma}{Lemma}

\begin{document}
\title{On the Dynamic Centralized Coded Caching Design}
\author{Qiaoling~Zhang, Lei~Zheng, Minquan~Cheng, Qingchun~Chen}
\maketitle

\begin{abstract}
Coded caching scheme provides us an effective framework to realize additional coded multicasting gain by exploiting coding into multiple transmitted signals. The goal of coded caching design is to jointly optimize the placement and delivery scheme so as to minimize the amount of transmitted coded signals. However, few research efforts consider multiple-round coded caching design problem, in which fixed and mobile users may coexist in one network different number of active users present in successive rounds. Obviously, such dynamic network configurations may lead to undesired frequent placement caching content updating at user sides, if we assume coded caching scheme for each round separately. Thus how to tailor the coded caching design, such that the frequent caching content updating in the placement phase can be avoided, and simultaneously retaining full caching gain in delivery phase in multiple rounds will become highly desirable. In this paper, by carefully examining the bipartite graph representation of the coded caching scheme,
a dynamic centralized coded caching design is proposed on the basis of the concatenating-based placement and the saturating matching based delivery scheme. Our analysis unveils that, the proposed dynamic coded caching scheme can realize the flexible coded multicast, and is order optimal.

\end{abstract}

\section{Introduction}
\IEEEPARstart {T}{he} global mobile data traffics are predicted to be increased sevenfold by 2021, in which over three-forths will be multimedia \cite{cisco}. The ever-increasing mobile data traffic has imposed huge burden to the networks. Recently, the cheaper cache with increasing availability provides us an alternative method to accommodate the explosive data traffic. In fact, by prefetching video contents at the end users, those locally cached contents can be directly served once they are requested, thus with caching the data traffic can be reduced, and this saving is referred to as the local caching gain.

In order to fully exploit the potential benefit of the cache, Maddah-Ali and Niesen proposed a $(K;M;N)$ centralized coded caching scheme (referred to as MN scheme for brevity) in \cite{MN}, in which a single server containing $N$ files with equal length coordinates $K$ users over a shared link, and every user is assumed to be provisioned with an identical cache of size $M$ files \cite{MN}. The coded caching scheme consists of two independent phases, \emph{namely,} placement phase and delivery phase, which is referred to as one \emph{round} in this paper. The placement phase occurs at off-peak hours, in which each user is able to access the server to fulfill its cache without the knowledge of users' demands. If user only prefetches the portions of the files at server, it is called uncoded placement. If the user fulfills its cache with some linear combinations of sub-packets from multiple files, it is the coded placement. Delivery phase follows in peak hours when users' demands are revealed. The server designs and multicasts coded messages through error-free and shared links to a set of users simultaneously, in which the global multicasting gain is maintained. By the end of the delivery phase, each user reconstructs its requested content on the basis of the received coded messages and its own caching contents. For this coded caching scheme, by jointly optimizing the placement and delivery phase, the system traffic load is expected to be minimized for all possible demands.

Motivated by the MN scheme, how to further reduce the required transmission load has attracted many research attentions.
An improved lower bound of the transmission load was derived from the combinatorial problem by optimally labeling the leaves of a directed tree in \cite{GR}. By interference elimination, a new scheme with smaller transmission load was disclosed for the case $K\geq N$ in \cite{T}.
More generally,  the transmission load for various demand patterns was derived by modifying the delivery phase of the MN scheme in \cite{YMA}. Note that  it was shown in \cite{WTP} and \cite{JCLC} that the MN scheme can achieve the minimum transmission load via graph theory and an optimization framework under a specific uncoded placement rule when $K< N$. Moreover, the MN scheme has been extended to different scenario of networks, for instance, the  multi-server systems \cite{Shariatpanahi_16,Mital_17}, D2D networks \cite{Ji_16}, hierarchical networks \cite{Karamchandani_16}, combination networks \cite{Ji_15}, and heterogenous network \cite{Daniel17}.

It should be addressed that, all the aforementioned works considered the coded caching scheme design within one round, namely, there is only one placement-then-delivery operation. Nonetheless, in practical applications, the coded caching system should be devised to operate within multiple rounds, in which the number of users $K$ may be time varying.
For instance, residents (fixed users) and visiting guests (mobile users) may coexist in one network. Intuitively, the residents may stay in network for a long time (multiple rounds)\footnote{If some fixed users request the same file in previous round, the fixed users can be removed from the coded caching design since their traffic requirements have been fulfilled.}, while those visiting guests may dynamically move in or out in different round of coded caching operation. For such a dynamic network, when applying the coded caching scheme to all the users at each round separately, the variations in the participating users may lead to frequent update in both the content caching and the signal transmission in order to make both the placement and the delivery fit the variations. Sometimes, this may become undesirable
and resource inefficient, especially when most of the users are fixed while only few users join or leave.
In this dynamic setup with multiple rounds of service request, how to tailor the coded caching design, such that the content updating in placement phase will be minimized and the full caching gain in delivery phase can be retained, will be a very interesting problem. This is exactly the motivation of our work in this paper.

In order to effectively handle the dynamic coded caching requirement, we need to rethink the content caching in placement phase and the coded signal generation in delivery phase for multiple rounds, such that the content updating at those fixed users can be minimized, while the coded caching gain for all participating users in delivery phase can be maximized. Intuitively the more users join the network in the same round, the possibility of the larger coding gain achieved should be higher. Therefore, all the set of fixed and mobile users should be considered when we design a coded caching scheme. However, in practice we do not have any knowledge of mobile users in the forthcoming rounds. To handle this issue, in this paper, we propose a Concatenating based placement and the Saturating Matching based delivery design (CSM) without the knowledge of mobile users. In the placement, the concatenating method is involved, in fact it has been widely used to cope with asynchronous problems. With this concatenating method we can keep the cache content unchanged for those fixed users who have already participated in the previous round of coded cooperation. For those newly joined mobile users, the server only needs to decide on the cache content placement by further sub-dividing the packets utilized by the fixed users. In this way, we can minimize the amount of the content updating. Since the matching over bipartite graph allows us to get the sum of the coded multicast transmissions from different groups (i.e., fixed and mobile users). Motivated by this, the saturating matching based delivery scheme is proposed. Our analysis reveals that the proposed CSM coded caching scheme is order-optimal.

The rest of the paper is organized as follows. In Section \ref{sec-pre}, the system model and some results of the original coded caching system in \cite{MN} are reviewed to introduce the $(K_1,K_2;M_1,M_2;N)$ dynamic coded caching design problem. Then the proposed CSM coded caching scheme and its order-optimality are presented in Section \ref{sec-the main results-1} and \ref{sec-first-performance}, respectively. Finally, we conclude our work in Section \ref{sec-conclusion}.

\section{System Model and Problem Formulation}\label{sec-pre}
\subsection{The Centralized Coded Caching Model}
Let us consider the centralized coded caching system (Fig. \ref{fig-origin-system}),
\begin{figure}[h]
\centering
\includegraphics[width=3.5in]{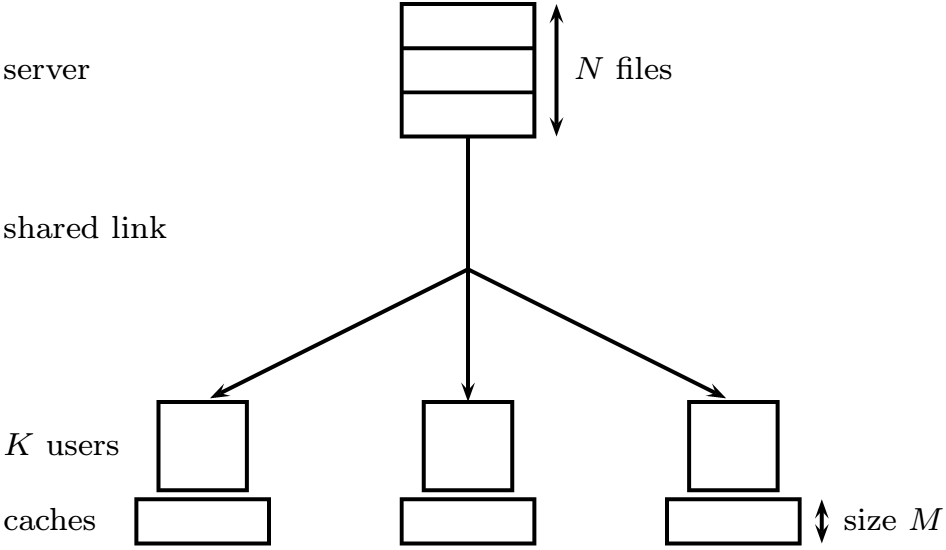}
\vskip 0.2cm
\caption{$(K;M;N)$ coded caching system}\label{fig-origin-system}
\end{figure}in which a server containing $N$ files denoted by $\mathcal{W}=\{W_{0}, W_{1}, \ldots, $ $W_{N-1}\}$ connects through an error-free shared link to $K$ users $\mathcal{K}=\{0,1,\ldots,K-1\}$ with $K<N$, and every user has a cache of size $M$ files for $M\in[0,N]$. The system contains two independent phases,
\begin{itemize}
  \item Placement phase: each file is divided into $F$ packets of equal size, and then  each user $k\in\mathcal{K}$ caches some packets out of each file, which is limited by its cache size $M$. 
  Let $\mathcal{Z}_k$ denote the cache contents at user $k$, which is assumed to be known to the server.
  \item Delivery phase: each user randomly requests one file from the server. The requested file by user $k$ is represented by $W_{d_k}$, and the request vector by all users is denoted by $\mathbf{d}=(d_0,d_1,\ldots,d_{K-1})$. 
 By the caching contents and requests from all the users, the server transmits some coded messages to all users such that each user's request can be satisfied.
\end{itemize}
In such system, the amount of worst-case transmissions for all possible requests is expected to be as small as possible, which is defined as
\begin{eqnarray*}
R=\sup_{\tiny{\mbox{$\begin{array}{c}
                 \mathbf{d}\in[N]^K\\
               \end{array}
$}}} \left\{\frac{S_{\mathbf{d}}}{F}\right\}
\end{eqnarray*} where $S_{\mathbf{d}}$ is the size of transmitted messages in the delivery phase for any given request $\mathbf{d}$.
The above $(K;M;N)$ coded caching system was firstly proposed in \cite{MN}, and its performance is given in Lemma \ref{le-MN}.
\begin{lemma}\rm(MN scheme \cite{MN})
\label{le-MN}
For any positive integers $K$, $M$ and $N$ with $t=\frac{KM}{N}\in \mathbb{N}^+$, there exists a $(K;M;N)$ scheme with transmission rate $R_{MN}(K;M;N)=\frac{K-t}{t+1}$.
\end{lemma}
For better understanding, the sketch of MN scheme is depicted by Algorithm \ref{alg-MN}, in which placement and delivery phase are included.
\begin{algorithm}[htb]
\caption{MN Scheme \cite{MN}}\label{alg-MN}
\begin{algorithmic}[1]
\Procedure {Placement}{$\mathcal{K}$, $\mathcal{W}$}
\State $t\leftarrow\frac{KM}{N}$
\State $\mathfrak{T}\leftarrow\{\mathcal{T}| \mathcal{T}\subset \mathcal{K},|\mathcal{T}|=t\}$
\For{$n\in[0,N)$}
\State Split $W_n$ into $W_n=\{W_{n,\mathcal{T}}|\mathcal{T}\in\mathfrak{T}\}$ of equal sized packet
\EndFor
\For{$k\in\mathcal{K}$}
\State $\mathcal{Z}_k\leftarrow\{W_{n,\mathcal{T}}|n\in[0,N),\mathcal{T}\in\mathfrak{T},k\in\mathcal{T}\}$
\EndFor
\EndProcedure
\Procedure{Delivery}{$\mathcal{W}$, $\mathbf{d}$}
\State $\mathfrak{S}\leftarrow\{\mathcal{S}|\mathcal{S}\subset\mathcal{K},|\mathcal{S}|=t+1\}$
\State Server sends $\{\oplus_{k\in\mathcal{S}} W_{d_k,\mathcal{S}\backslash \{k\}}|\mathcal{S}\in\mathfrak{S}\}$
\EndProcedure
\end{algorithmic}
\end{algorithm}

From Algorithm \ref{alg-MN}, it is clear that each file is divided into $\binom{K}{t}$ nonoverlapping equal-sized packets, and for a given $t$, there are in total $\binom{K}{t+1}$ coded messages. To sum up, the transmission rate of MN scheme can thus be derived as,
\begin{eqnarray*}
\label{eq-MN-rate}
R_{MN}(K;M;N)=\frac{{K\choose t+1}}{{K\choose t}}=\frac{K-t}{t+1}.
\end{eqnarray*}

\subsection{Dynamic Coded Caching Problem Formulation}
As illustrated in Fig.\ref{Fig-sys}, now we consider a similar network configuration as that in the above $(K;M;N)$ model, except the fact that there are two sets of users, namely, $K_1$ fixed users $\mathcal{K}_1=\{0,1,\ldots,K_1-1\}$, each of which has a cache of size $M_{1}$ files, and $K_2$ mobile users $\mathcal{K}_2=\{K_1,K_1+1,\ldots\}$, each of which has a cache of size $M_{2}$ files.
It should be noted that, the cache sizes of two user sets $\mathcal{K}_1$ and $\mathcal{K}_2$ are not necessarily the same. For notation brevity, throughout the paper, we refer to this model as $(K_1,K_2;M_1,M_2;N)$ coded caching system.
\begin{figure}[h]
\centering
\includegraphics[width=4in]{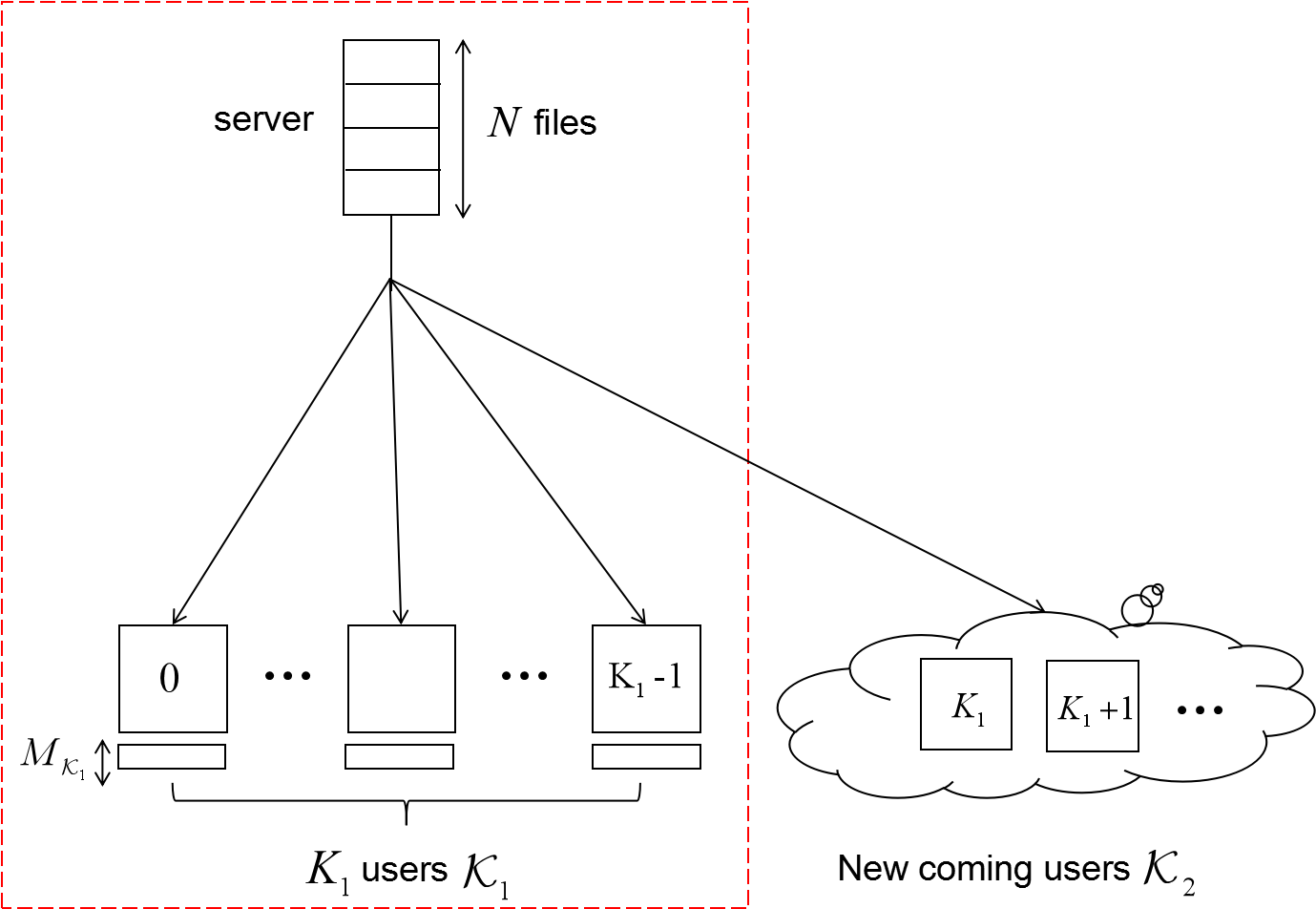}
\vskip 0.2cm
\caption{$(K_1,K_2;M_{1},M_{2};N)$ coded caching model}\label{Fig-sys}
\end{figure}
The above $(K_1,K_2;M_1,M_2;N)$ coded caching system can be utilized to characterize the dynamic network at any round $i$, in which $K_2$ mobile users join this system and want to perform coded caching with the $K_1$ fixed users that are already in the system. The first problem is related to the placement design. We aim to design a placement scheme intended only for the $K_2$ users $\mathcal{K}_2$ in order to minimize the caching content update, which makes sense when there is no change in the files at server, and those $K_1$ fixed users have already filled their caches in previous rounds\footnote{Note that if some new users move into network during round $i$, their cache will be filled in placement phase of the following round, \emph{i.e.}, round $i+1$.}. On the basis of designed caching contents for users in $\mathcal{K}_2$, how to design the coded multicasting among all the $K_1+K_2$ users to minimize the required transmission load will be the second design objective.

In order to derive the placement and delivery scheme suitable for the aforementioned dynamic coded caching applications, the bipartite graph representation of coded caching will be utilized, and the following symbol notations are utilized in the following analysis. Basically, a graph can be denoted by $\mathbf{G}=(\mathcal{V},\mathcal{E})$, where $\mathcal{V}$ is the set of vertices and $\mathcal{E}$ is the set of edges, and a subset of edges $\mathcal{M}\subseteq \mathcal{E}$ is a matching if no two edges have a common vertex. A bipartite graph, denoted by $\mathbf{G}=(\mathcal{X},\mathcal{Y}; \mathcal{E})$, is a graph whose vertices are divided into two disjoint parts $\mathcal{X}$ and $\mathcal{Y}$ such that every edge in $\mathcal{E}$ connects a vertex in $\mathcal{X}$ to one in $\mathcal{Y}$. For a set $X\subseteq \mathcal{X}$, let $N_{\mathbf{G}}(X)$ denote the set of all vertices in $\mathcal{Y}$ adjacent to some vertices of $X$. The degree of a vertex is defined as the number of vertices adjacent to it. If every vertex of $\mathcal{X}$ has the same degree, we  call it the degree of $\mathcal{X}$ and denote as $d(\mathcal{X})$.

\section{The Concatenating based Placement and the Saturating Matching based Delivery Scheme}\label{sec-the main results-1}
In this section, we mainly focus on the $(K_1,K_2;M_1,M_2;N)$ coded caching design, and its achieved transmission rate for any $(M_1,M_2)$ pairs. Before introducing our proposed scheme, intuitively, the original MN scheme can be simply applied into our considered system by regarding $\mathcal{K}_1$ and $\mathcal{K}_2$ as two separate user groups. Thus, for the convenience of comparison, we firstly give the following trivial examples when directly utilizing MN scheme.

\subsection{Baseline Scheme}
Intuitively, when adopting MN scheme without updating the cache contents of users in $\mathcal{K}_1$, all the $K_1+K_2$ users can be classified into two groups, $\mathcal{K}_1$ and  $\mathcal{K}_2$, and perform MN coded caching scheme for two groups separately, which will be introduced in the following Example \ref{ex-trivial}.
\begin{example}\rm
\label{ex-trivial}
Consider a $(K_1,K_2;M_1,M_2;N)$ system consisting of $K_1=4$ fixed users $\mathcal{K}_1=\{0,1,2,3\}$, $K_2=3$ mobile user $\mathcal{K}_2=\{4,5,6\}$ and $N=7$ files $\mathcal{W}=\{W_0,W_1,\ldots,W_{6}\}$. When assume that  $M_{1}=\frac{21}{4}$, $M_{2}=\frac{14}{3}$,  we have $t_1=3$, and $t_2=2$, respectively.

Firstly, by applying the placement scheme in Algorithm \ref{alg-MN} to all users in $\mathcal{K}_1$, each file is split into $4$ packets of equal size, \emph{i.e.},
\begin{eqnarray}
\label{exm-K=2-2}
W_n=\{W_{n,\{0,1,2\}}, W_{n,\{0,1,3\}}, W_{n,\{0,2,3\}}, W_{n,\{1,2,3\}}\}, \  n=0,1,\ldots,6.
\end{eqnarray}
and the four fixed users in $ \mathcal{K}_1$ will cache the following contents in placement phase,
\begin{eqnarray}
\label{ex-k1-mn-p}
\begin{split}
\mathcal{Z}_{0}&=\{ W_{n,\{0,1,2\}}, W_{n,\{0,1,3\}}, W_{n,\{0,2,3\}}\}_{n=0}^6,\ \
\mathcal{Z}_{1}=\{ W_{n,\{0,1,2\}}, W_{n,\{0,1,3\}}, W_{n,\{1,2,3\}}\}_{n=0}^6,\\
\mathcal{Z}_{2}&=\{ W_{n,\{0,1,2\}}, W_{n,\{0,2,3\}}, W_{n,\{1,2,3\}}\}_{n=0}^6,\ \
\mathcal{Z}_{3}=\{ W_{n,\{0,1,3\}}, W_{n,\{0,2,3\}}, W_{n,\{1,2,3\}}\}_{n=0}^6.
\end{split}
\end{eqnarray}
Similarly for the three users in $\mathcal{K}_2$, each file will be split into $3$ packets of the same size, \emph{i.e.},
\begin{eqnarray*}
W_n=\{W_{n,\{4,5\}},  W_{n,\{4,6\}}, W_{n,\{5,6\}}\},\  n=0,1,\ldots,6
\end{eqnarray*}
and the users in $\mathcal{K}_2$ will cache the following contents in placement phase,
\begin{eqnarray*}
\mathcal{Z}_{4}=\{ W_{n,\{4,5\}}, W_{n,\{4,6\}}\}_{n=0}^6,
\mathcal{Z}_{5}=\{ W_{n,\{4,5\}}, W_{n,\{5,6\}}\}_{n=0}^6,
\mathcal{Z}_{6}=\{ W_{n,\{4,6\}}, W_{n,\{5,6\}}\}_{n=0}^6.
\end{eqnarray*}

Without loss of generality, assume $\mathbf{d}=\{0,1,2,3,4,5,6\}$. By using the delivery scheme in Algorithm \ref{alg-MN}, we have $\mathcal{S}_1=\{0,1,2,3\}$ and $\mathcal{S}_2=\{4,5,6\}$, then the sever respectively sends
    $$\mathop{\oplus}_{k_1\in \mathcal{S}_1}W_{d_{k_1},\mathcal{S}_1\setminus\{k_1\}}\ \ \ \hbox{and}\ \ \ \mathop{\oplus}_{k_2\in \mathcal{S}_2}W_{d_{k_2},\mathcal{S}_2\setminus\{k_2\}}.$$
    to the users in $\mathcal{K}_1$ and $\mathcal{K}_2$.

By the Lemma \ref{le-MN}, we have $R_{MN}(4;\frac{21}{4};7)=\frac{1}{4}$ for $\mathcal{K}_1$ and $R_{MN}(3;\frac{14}{3};7)=\frac{1}{4}$ for $\mathcal{K}_2$ respectively. Then to sum up, the total transmission rate is $\frac{1}{4}+\frac{1}{3}=\frac{7}{12}$.
\end{example}

From the above Example \ref{exm-K=2-2}, it can be observed that, when directly applying MN scheme to the $(K_1,K_2;M_1,M_2;N)$ system, users need to be divided into two groups, and MN scheme $(K_1;M_1;N)$ and $(K_2;M_2;N)$ are utilized in each group, respectively. Although, this strategy is applicable, during delivery, the multicasting opportunities between users in $\mathcal{K}_1$ and $\mathcal{K}_2$ are lost, which is less efficient. At this point, how to design a scheme to maximally exploit the multicasting gain among all the users would be an interesting problem, and is worth to be investigagted.

\subsection{The Proposed Dynamic Coded Caching Design and The Main Results}
Unlike the application of the MN scheme to two user sets separately, for the $(K_1,K_2;M_{1},M_{2};N)$ caching system, we propose a new design by concatenating two user groups, such that the coding gain can be enlarged. More specifically, the coded multicasting gain in our proposed scheme is $\frac{M_1K_1}{N}+\frac{M_2K_2}{N}$. When $M_1=M_2$, our coded gain is only one gain less than the maximum gain of the MN scheme $(K_1+K_2;M_1;N)$. The transmission rate of our proposed scheme is given by Theorem \ref{th-new-r}.
\begin{theorem}\rm
\label{th-new-r}
For the positive integers $K_i$, $N$ and $M_{i}<N$ such that $t_i=\frac{K_iM_{i}}{N}\in [1,K_i)$, $i=1,2$, there exists a  $(K_1,K_2;M_{1},M_{2};N)$ coded caching scheme to achieve the following transmission rate
\begin{eqnarray*}
R(K_1,K_2;M_1,M_2;N)\leq\begin{cases}
\frac{(K_1-t_1)(t_1t_2+t_1+1)}{(t_1+1)t_1t_2}+\frac{(K_2-t_2)}{(t_1+1)t_2},&\text{if $M_1<M_2$},\\[0.2cm]
\frac{(K_1-t_1)}{(t_2+1)t_1}+\frac{(K_2-t_2)(t_1t_2+t_2+1)}{(t_2+1)t_1t_2},&\text{if $M_1\geq M_2$}.
\end{cases}
\end{eqnarray*}
\end{theorem}

The proof of Theorem \ref{th-new-r} is presented in Appendix. The placement and the delivery scheme of the proposed CSM coded caching scheme are given in Algorithm \ref{alg-NSP} and Algorithm \ref{alg-NSD} respectively.
\begin{small}
\begin{algorithm}[htp]
\caption{Concatenating Placement Design of The $(K_1,K_2;M_1,M_2;N)$ Dynamic Coded Caching Scheme}\label{alg-NSP}
\begin{algorithmic}[1]
\Procedure {Placement 1 } {$\mathcal{K}_1$, $M_1$, $\mathcal{W}$}  
\State $t_1 \leftarrow \frac{M_{1}K_1}{N}$
\State $\mathfrak{A} \leftarrow \{\mathcal{A}|\mathcal{A}\subset \mathcal{K}_1, |\mathcal{A}|=t_1\}$
\State Split each file $W_i\in\mathcal{W}$ into ${K_1 \choose t_1}$ packets, \emph{i.e.}, $W_{i}=\{W_{i,\mathcal{A}}\ |\ \mathcal{A}\in \mathfrak{A}\}$.
\For{$k_1\in \mathcal{K}_1$}
\State $\mathcal{Z}_{k_1}\leftarrow\left \{W_{i,\mathcal{A}}\ |\ k_1\in\mathcal{A},\mathcal{A}\in\mathfrak{A},i\in[0,N)\right\}$
\EndFor
\EndProcedure
\Procedure {Placement 2 }{$\mathcal{K}_2$, $M_2$, $\mathcal{W}$}
\State $t_2 \leftarrow \frac{M_{2}K_2}{N}$
\State $\mathfrak{B}\leftarrow \{\mathcal{B}|\mathcal{B}\subset \mathcal{K}_2, |\mathcal{B}|=t_2\}$
\State Split each packet $W_{i,\mathcal{A}}$ into $t_1t_2{K_2 \choose t_2}$ sub-packets, \emph{i.e.}, $$W_{i,\mathcal{A}}=\left \{W_{i,\mathcal{A}, \mathcal{B}}^{(l_2,l_1)}\ |\ \mathcal{A}\in \mathfrak{A}, \mathcal{B}\in\mathfrak{B}, l_1\in[0,t_1), l_2\in[0,t_2)\right\}.$$
\For{$k_2\in \mathcal{K}_2$}
\State $\mathcal{Z}_{k_2}\leftarrow\left \{W_{i,\mathcal{A},\mathcal{B}}^{(l_2,l_1)}\ |\ k_2\in\mathcal{B}, \mathcal{B}\in\mathfrak{B}, \mathcal{A}\in\mathfrak{A}, i\in[0,N), l_1\in[0,t_1), l_2\in[0,t_2)\right\}$
\EndFor
\EndProcedure
\end{algorithmic}
\end{algorithm}
\end{small}

In Algorithm \ref{alg-NSP}, the PLACMENT 1 is for the users in $\mathcal{K}_1$, which may be performed in any previous rounds, and the PLACMENT 2 is designed for those mobile users in $\mathcal{K}_2$. It can be observed that PLACMENT 1 is independent of PLACMENT 2, while that for users in $\mathcal{K}_2$ depends on the number of users in $\mathcal{K}_1$ to maximally utilize the coded multicasting opportunities among all the users. After this placement phase, the cache contents of all the users denoted by $\mathcal{Z}$ are assumed to be known by the server. Then, the delivery phase follows and is described in Algorithm \ref{alg-NSD}.

\begin{small}
\begin{algorithm}[htp]
\caption{The Saturating Matching Delivery Design for The $(K_1,K_2;M_1,M_2;N)$ Dynamic Coded Caching Scheme}\label{alg-NSD}
\begin{algorithmic}[1]
 \Procedure{Find Vertices of Bipartite Graph }{$\mathcal{Z}$, ${\bf d}$, $\mathcal{W}$}
\State $\mathfrak{S}_1\leftarrow \{\mathcal{S}_1| \mathcal{S}_1\subseteq \mathcal{K}_1, |\mathcal{S}_1|=t_1+1\}$
\State $\mathfrak{S}_2\leftarrow \{\mathcal{S}_2| \mathcal{S}_2\subseteq \mathcal{K}_2, |\mathcal{S}_2|=t_2+1\}$
\For{$\mathcal{S}_1=\{k_{1,0},k_{1,1},\ldots,k_{1,t_1}\}\in \mathfrak{S}_1$, $\mathcal{B}\in \mathfrak{B}$}
\For{$l_2\in [0,t_2)$, $m_1\in [0,t_1]$}
\State
{\begin{align}
X^{(l_2,m_1)}_{\mathcal{S}_{1},\mathcal{B}}\triangleq
\left(\bigoplus\limits_{m'_1=0}^{m_1-1} W^{(l_2,m_1-1)}_{d_{k_{1,m_1'}},\mathcal{S}_1\setminus\{k_{1,m'_1}\},\mathcal{B}}\right)
\bigoplus
\left(
\bigoplus\limits_{m'_1=m_1+1}^{t_1} W^{(l_2,m_1)}_{d_{k_{1,m_1'}},\mathcal{S}_1\setminus\{k_{1,m'_1}\},\mathcal{B}}\right)
\end{align}\label{eq-ex-step1-0}}
\EndFor
\EndFor

\For{$\mathcal{S}_2=\{k_{2,0},k_{2,1},\ldots,k_{2,t_1}\}\in \mathfrak{S}_2$, $\mathcal{A}\in \mathfrak{A}$ }
\For{ $l_1\in[0,t_1)$, $m_2\in [0,t_2]$ }
\begin{eqnarray}
\label{eq-ex-step1-1}
Y^{(m_2,l_1)}_{\mathcal{A},\mathcal{S}_2 }\triangleq \left(\bigoplus\limits_{m'_2=0}^{m_2-1} W^{(m_2-1,l_1)}_{d_{k_{2,m'_2}},\mathcal{A},\mathcal{S}_2\setminus\{k_{2,m'_2}\}}\right)\bigoplus\left(
\bigoplus\limits_{m'_2=m_2+1}^{t_2} W^{(m_2,l_1)}_{d_{k_{2,m'_2}},\mathcal{A},\mathcal{S}_2\setminus\{k_{2,m'_2}\}}\right)
\end{eqnarray}
\EndFor
\EndFor
\EndProcedure
\Procedure{Define Bipartite Graph $\mathbf{G}=(\mathcal{X},\mathcal{Y}; \mathcal{E})$}{}
\State { Define $\mathbf{G}=(\mathcal{X},\mathcal{Y}; \mathcal{E})$ such that $X^{(l_2,m_1)}_{\mathcal{S}_{1},\mathcal{B}}\in \mathcal{X}$ is adjacent to $Y^{(m_2,l_1)}_{\mathcal{A},\mathcal{S}_2 }\in \mathcal{Y}$ if and only if $\mathcal{S}_{1}\setminus\{k_{1,m_1}\}=\mathcal{A}$ and $\mathcal{S}_{2}\setminus\{k_{2,m_2}\}=\mathcal{B}$, where}\label{line:req}
 $$\mathcal{X}=\left\{X^{(l_2,m_1)}_{\mathcal{S}_{1},\mathcal{B}}\ |\ \mathcal{S}_1\in \mathfrak{S}_1, \mathcal{B}\in \mathfrak{B},l_2\in [0,t_2), m_1\in [0,t_1]\right\}$$
$$\mathcal{Y}=\left\{Y^{(m_2,l_1)}_{\mathcal{A},\mathcal{S}_2 }\ |\ \mathcal{S}_2\in \mathfrak{S}_2, \mathcal{A}\in \mathfrak{A},l_1\in [0,t_1), m_2\in [0,t_2]\right\}$$
\EndProcedure
\Procedure {Delivery }{$\mathbf{G}=(\mathcal{X},\mathcal{Y}; \mathcal{E})$}
\State {Find an appropriate saturating matching of $\mathbf{G}$ for $\mathcal{X}$ (or $\mathcal{Y}$), denoted by $\mathbf{G}''=(\mathcal{X},\mathcal{Y}''; \mathcal{E}'')$, such that the remaining vertices in $\mathcal{Y}$ or ($\mathcal{X}$) as in \eqref{eq-ex-step1-1} can be sent with multicasting gain $t_2+1$ ($t_1+1$) as much as possible}\label{line:match}
\For { each edge $e=\left(X^{(l_2,m_1)}_{\mathcal{S}_{1},\mathcal{B}},Y^{(m_2,l_1)}_{\mathcal{A},\mathcal{S}_2 }\right)\in \mathcal{E}''$}
\State Sends
$X^{(l_2,m_1)}_{\mathcal{S}_{1},\mathcal{B}}\bigoplus Y^{(m_2,l_1)}_{\mathcal{A},\mathcal{S}_2 }$\label{line:send}
\EndFor
\State Sends remaining vertices in $\mathcal{Y}$ ($\mathcal{X}$) with $t_2$ ($t_1$) for every $t_2+1$ ($t_1+1$) vertices.
\EndProcedure
\end{algorithmic}
\end{algorithm}
\end{small}

From the Algorithm \ref{alg-NSD}, it can be observed that, the whole delivery phase can be realized by firstly finding the vertices of the bipartite graph, and then constructing the bipartite graph according to the requirements described in Line \ref{line:req}. Finally based on the constructed graph we can find out the coded messages as described in Procedure Delivery. The working flow of the Algorithm \ref{alg-NSD} can be briefly summarized as below: in Lines 2-13, we try to generate the vertices of a bipartite graph by means of the required sub-packets, i.e., to create the multicasting gains $\frac{M_1K_1}{N}$ and $\frac{M_2K_2}{N}$ respectively;  in Line \ref{line:req} we can define its edge, i.e., to create the total coded multicasting gain $\frac{M_1K_1}{N}+\frac{M_2K_2}{N}$ according to the requests of users; Finally in Procedure Delivery, based on the saturating matching, the messages with multicasting gain $\frac{M_1K_1}{N}+\frac{M_2K_2}{N}$ are broadcast as much as possible.

It is worth noting that, in the Proposition \ref{pro-baohe-X} in APPENDIX, we show that there always exists a saturating matching of $\mathbf{G}$ for $\mathcal{X}$ (or $\mathcal{Y}$), denoted by $\mathbf{G}''=(\mathcal{X},\mathcal{Y}''; \mathcal{E}'')$, such that the remaining sub-packets in $\mathcal{Y}$ (or $\mathcal{X}$) as in \eqref{eq-ex-step1-1} can be sent with multicasting gain of $t_2+1$ ($t_1+1$) as much as possible, i.e., Line 18-22 work.
\subsection{The Illustration of The Proposed Dynamic Coded Caching Design}
For clear illustration of our proposed scheme, in this subsection, we give the following example to illustrate the Algorithm  \ref{alg-NSP} and Algorithm \ref{alg-NSD}.
\begin{example}\rm
\label{ex-neq}
To compare with the baseline scheme in Example \ref{ex-trivial}, we assume the same system setup here. According to Algorithm \ref{alg-NSP} and Algorithm \ref{alg-NSD}, the placement and delivery phase are explicated as follows. It needs to be highlighted that, the placement for $\mathcal{K}_1$ and $\mathcal{K}_2$ are assumed to be performed in different round, just like we have addressed before.
\begin{itemize}
\item {\bf Placement phase}: The placement phase can be completed by the following two steps.
\begin{itemize}
\item Step 1 (Placement for $\mathcal{K}_1$): At any previous round,  the placement phase in Algorithm \ref{alg-MN} is applied to users in $\mathcal{K}_1$,which is the same as \eqref{ex-k1-mn-p} in Example \ref{ex-trivial}.
\item Step 2 (Placement for $\mathcal{K}_2$): By Lines $10-12$ in Algorithm \ref{alg-NSP}, for each packet in \eqref{exm-K=2-2}, we further divide it into $t_1t_2{K_2 \choose t_2}=3\cdot 2 \cdot{3\choose 2}=18$ sub-packets with equal size, \emph{i.e.,} for each $n=0,1,\ldots,6$,
\begin{equation*}
\begin{split}
W_{n,\{0,1,2\} }&=\left\{W_{n,\{0,1,2\},\mathcal{B}}^{(l_2,l_1)}\ |\ \mathcal{B}\subset \mathcal{K}_2, |\mathcal{B}|=2, \ l_1\in [0,3),\ l_2\in[0,2)\right\},\\
W_{n,\{0,1,3\} }&=\left\{W_{n,\{0,1,3\},\mathcal{B}}^{(l_2,l_1)}\ |\ \mathcal{B}\subset \mathcal{K}_2, |\mathcal{B}|=2, \ l_1\in [0,3),\ l_2\in[0,2)\right\},\\
W_{n,\{0,2,3\} }&=\left\{W_{n,\{0,2,3\},\mathcal{B}}^{(l_2,l_1)}\ |\ \mathcal{B}\subset \mathcal{K}_2, |\mathcal{B}|=2, \ l_1\in [0,3),\ l_2\in[0,2)\right\},\\
W_{n,\{1,2,3\} }&=\left\{W_{n,\{1,2,3\},\mathcal{B}}^{(l_2,l_1)}\ |\ \mathcal{B}\subset \mathcal{K}_2, |\mathcal{B}|=2, \ l_1\in [0,3),\ l_2\in[0,2)\right\}.
\end{split}
\end{equation*}
By Lines $14$ in Algorithm \ref{alg-NSP}, each users $k_2\in \mathcal{K}_2$ caches
$$\mathcal{Z}_{k_2}=\left\{W^{(l_2,l_1)}_{n,\mathcal{A},\mathcal{B}}\ |\ k_2\in \mathcal{B}, \mathcal{A}\subset{\mathcal{K}_1},\mathcal{B}\subset{\mathcal{K}_2}, |\mathcal{A}|=3, |\mathcal{B}|=2, l_1\in [0,3), l_2\in[0,2)\right\}.$$
After this step, each user $k_2$ caches a total of $t_1t_2{K_2-1 \choose t_2-1}=3\cdot2\cdot{2\choose 1}=12$ sub-packets from each packet. Since each file is firstly divided into $\binom{4}{3}$ packets, each
  user $k_2$ caches $12\cdot {4\choose 3}\cdot N=12\cdot 4\cdot N=48N$ sub-packets. Since the total number of sub-packets in each file is $18\cdot{4\choose 3}=72$, the user $k_2$ caches $\frac{48N}{72}=\frac{2}{3}N=\frac{14}{3}$ files, which satisfies its cache size limitation.
\end{itemize}
\item {\bf Delivery phase}: Also assume that $\mathbf{d}=\{0,1,2,3,4,5,6\}$, and from the caching result in Algorithm \ref{alg-NSP},
we have  $$\mathcal{S}_1=\{0,1,2,3\},\ \ \mathcal{A}_0=\{1,2,3\},\ \ \mathcal{A}_1=\{0,2,3\},\ \ \mathcal{A}_2=\{0,1,3\},\ \ \mathcal{A}_3=\{0,1,2\}$$  $$\mathcal{S}_2=\{4,5,6\},\ \ \mathcal{B}_0=\{5,6\},\ \ \mathcal{B}_1=\{4,6\},\ \ \mathcal{B}_2=\{4,5\}.$$
   By the first procedure in Algorithm \ref{alg-NSD}, we have the following vertices, where the elements in the set of $\mathcal{X}$ and $\mathcal{Y}$ are labeled as $u_i, (i=0,\ldots,23)$, and $v_j, (j=0,\ldots,35)$ in short respectively.
\begin{small}
\begin{eqnarray*}
\begin{array}{cccccc}
  u_0=X^{(0,0)}_{\mathcal{S}_1,\mathcal{B}_0}
& u_1=X^{(0,1)}_{\mathcal{S}_1,\mathcal{B}_0}
& u_2=X^{(0,2)}_{\mathcal{S}_1,\mathcal{B}_0}
& u_3=X^{(0,3)}_{\mathcal{S}_1,\mathcal{B}_0}
& u_4=X^{(0,0)}_{\mathcal{S}_1,\mathcal{B}_1}
& u_5=X^{(0,1)}_{\mathcal{S}_1,\mathcal{B}_1}\\[0.2cm]
  u_6=X^{(0,2)}_{\mathcal{S}_1,\mathcal{B}_1}
& u_7=X^{(0,3)}_{\mathcal{S}_1,\mathcal{B}_1}
& u_8=X^{(0,0)}_{\mathcal{S}_1,\mathcal{B}_2}
& u_9=X^{(0,1)}_{\mathcal{S}_1,\mathcal{B}_2}
& u_{10}=X^{(0,2)}_{\mathcal{S}_1,\mathcal{B}_2}
& u_{11}=X^{(0,3)}_{\mathcal{S}_1,\mathcal{B}_2}\\[0.2cm]
  u_{12}=X^{(1,0)}_{\mathcal{S}_1,\mathcal{B}_0}
& u_{13}=X^{(1,1)}_{\mathcal{S}_1,\mathcal{B}_0}
& u_{14}=X^{(1,2)}_{\mathcal{S}_1,\mathcal{B}_0}
& u_{15}=X^{(1,3)}_{\mathcal{S}_1,\mathcal{B}_0}
& u_{16}=X^{(1,0)}_{\mathcal{S}_1,\mathcal{B}_1}
& u_{17}=X^{(1,1)}_{\mathcal{S}_1,\mathcal{B}_1}\\[0.2cm]
  u_{18}=X^{(1,2)}_{\mathcal{S}_1,\mathcal{B}_1}
& u_{19}=X^{(1,3)}_{\mathcal{S}_1,\mathcal{B}_1}
& u_{20}=X^{(1,0)}_{\mathcal{S}_1,\mathcal{B}_2}
& u_{21}=X^{(1,1)}_{\mathcal{S}_1,\mathcal{B}_2}
& u_{22}=X^{(1,2)}_{\mathcal{S}_1,\mathcal{B}_2}
& u_{23}=X^{(1,3)}_{\mathcal{S}_1,\mathcal{B}_2}\\[0.2cm]
  v_0=Y^{(0,0)}_{\mathcal{A}_0,\mathcal{S}_2} &v_1=Y^{(1,0)}_{\mathcal{A}_0,\mathcal{S}_2}
& v_2=Y^{(2,0)}_{\mathcal{A}_0,\mathcal{S}_2} &v_3=Y^{(0,0)}_{\mathcal{A}_1,\mathcal{S}_2}&
  v_4=Y^{(1,0)}_{\mathcal{A}_1,\mathcal{S}_2} &v_5=Y^{(2,0)}_{\mathcal{A}_1,\mathcal{S}_2}\\[0.2cm]
  v_6=Y^{(0,0)}_{\mathcal{A}_2,\mathcal{S}_2} &v_7=Y^{(1,0)}_{\mathcal{A}_2,\mathcal{S}_2}&
  v_8=Y^{(2,0)}_{\mathcal{A}_2,\mathcal{S}_2} &v_9=Y^{(0,0)}_{\mathcal{A}_3,\mathcal{S}_2}
& v_{10}=Y^{(1,0)}_{\mathcal{A}_3,\mathcal{S}_2}&v_{11}=Y^{(2,0)}_{\mathcal{A}_3,\mathcal{S}_2}\\[0.2cm]
  v_{12}=Y^{(0,1)}_{\mathcal{A}_0,\mathcal{S}_2}&v_{13}=Y^{(1,1)}_{\mathcal{A}_0,\mathcal{S}_2}
& v_{14}=Y^{(2,1)}_{\mathcal{A}_0,\mathcal{S}_2}&v_{15}=Y^{(0,1)}_{\mathcal{A}_1,\mathcal{S}_2}&
  v_{16}=Y^{(1,1)}_{\mathcal{A}_1,\mathcal{S}_2}&v_{17}=Y^{(2,1)}_{\mathcal{A}_1,\mathcal{S}_2}\\[0.2cm]
  v_{18}=Y^{(0,1)}_{\mathcal{A}_2,\mathcal{S}_2}&v_{19}=Y^{(1,1)}_{\mathcal{A}_2,\mathcal{S}_2}&
  v_{20}=Y^{(2,1)}_{\mathcal{A}_2,\mathcal{S}_2}&v_{21}=Y^{(0,1)}_{\mathcal{A}_3,\mathcal{S}_2}
& v_{22}=Y^{(1,1)}_{\mathcal{A}_3,\mathcal{S}_2}&v_{23}=Y^{(2,1)}_{\mathcal{A}_3,\mathcal{S}_2}\\[0.2cm]
  v_{24}=Y^{(0,2)}_{\mathcal{A}_0,\mathcal{S}_2}&v_{25}=Y^{(1,2)}_{\mathcal{A}_0,\mathcal{S}_2}
& v_{26}=Y^{(2,2)}_{\mathcal{A}_0,\mathcal{S}_2}&v_{27}=Y^{(0,2)}_{\mathcal{A}_1,\mathcal{S}_2}&
  v_{28}=Y^{(1,2)}_{\mathcal{A}_1,\mathcal{S}_2}&v_{29}=Y^{(2,2)}_{\mathcal{A}_1,\mathcal{S}_2}\\[0.2cm]
  v_{30}=Y^{(0,2)}_{\mathcal{A}_2,\mathcal{S}_2}&v_{31}=Y^{(1,2)}_{\mathcal{A}_2,\mathcal{S}_2}&
  v_{32}=Y^{(2,2)}_{\mathcal{A}_2,\mathcal{S}_2}&v_{33}=Y^{(0,2)}_{\mathcal{A}_3,\mathcal{S}_2}&
  v_{34}=Y^{(1,2)}_{\mathcal{A}_3,\mathcal{S}_2}&v_{35}=Y^{(2,2)}_{\mathcal{A}_3,\mathcal{S}_2}
\end{array}
\end{eqnarray*}
\end{small}

By the second prodecure in Algorithm \ref{alg-NSD}, the following bipartite graph $\mathbf{G}=(\mathcal{X},\mathcal{Y}; \mathcal{E})$ depicted in Figure \ref{Fig-exm-ebt} can be obtained.
\begin{figure}[h]
\centering\includegraphics[width=0.9\textwidth]{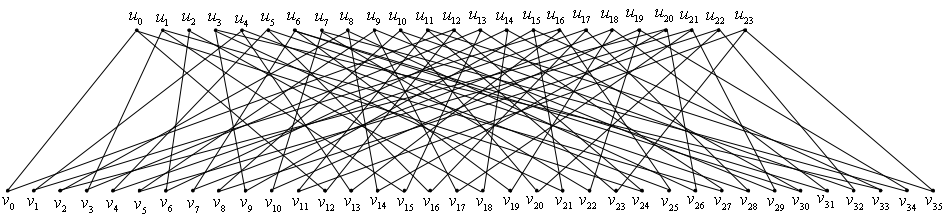}
\caption{The bipartite graph $\mathbf{G}=(\mathcal{X},\mathcal{Y}; \mathcal{E})$ in Example \ref{ex-neq}}\label{Fig-exm-ebt}
\end{figure}

Then, with this graph $\mathbf{G}$, we can find a saturating matching for $\mathcal{X}$ as  in Line \ref{line:match}  described in the last precedure of Algorithm \ref{alg-NSD}, which is depicted in Figure \ref{Fig-exm-bhpp}.
\begin{figure}[h]
\centering\includegraphics[width=0.9\textwidth]{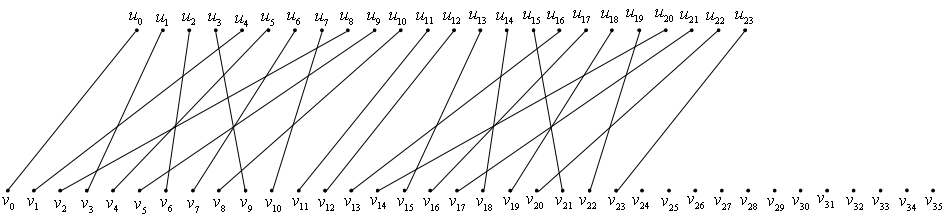}
\caption{A saturating matching for $\mathcal{X}$ in Example \ref{ex-neq}}\label{Fig-exm-bhpp}
\end{figure}
From the saturating matching and the Line \ref{line:send} in Algorithm \ref{alg-NSD}, the server sends the following $24$ coded messages involving sub-packets from both $\mathcal{X}$ and $\mathcal{Y}$.
\begin{small}
\begin{eqnarray*}
\begin{array}{ccc}
u_0\bigoplus v_0=X^{(0,0)}_{\mathcal{S}_1,\mathcal{B}_0}\bigoplus Y^{(0,0)}_{\mathcal{A}_0,\mathcal{S}_2}
& u_1\bigoplus v_3=X^{(0,1)}_{\mathcal{S}_1,\mathcal{B}_0}\bigoplus Y^{(0,0)}_{\mathcal{A}_1,\mathcal{S}_2}
& u_2\bigoplus v_6=X^{(0,2)}_{\mathcal{S}_1,\mathcal{B}_0}\bigoplus Y^{(0,0)}_{\mathcal{A}_2,\mathcal{S}_2}\\[0.2cm]
  u_3\bigoplus v_9=X^{(0,3)}_{\mathcal{S}_1,\mathcal{B}_0}\bigoplus Y^{(0,0)}_{\mathcal{A}_3,\mathcal{S}_2}
& u_4\bigoplus v_1=X^{(0,0)}_{\mathcal{S}_1,\mathcal{B}_1}\bigoplus Y^{(1,0)}_{\mathcal{A}_0,\mathcal{S}_2}
& u_5\bigoplus v_4=X^{(0,1)}_{\mathcal{S}_1,\mathcal{B}_1}\bigoplus Y^{(1,0)}_{\mathcal{A}_1,\mathcal{S}_2}\\[0.2cm]
  u_6\bigoplus v_7=X^{(0,2)}_{\mathcal{S}_1,\mathcal{B}_1}\bigoplus Y^{(1,0)}_{\mathcal{A}_2,\mathcal{S}_2}
& u_7\bigoplus v_{10}=X^{(0,3)}_{\mathcal{S}_1,\mathcal{B}_1}\bigoplus Y^{(1,0)}_{\mathcal{A}_3,\mathcal{S}_2}
& u_8\bigoplus v_{2}=X^{(0,0)}_{\mathcal{S}_1,\mathcal{B}_2}\bigoplus Y^{(2,0)}_{\mathcal{A}_0,\mathcal{S}_2}\\[0.2cm]
  u_9\bigoplus v_5=X^{(0,1)}_{\mathcal{S}_1,\mathcal{B}_2}\bigoplus Y^{(2,0)}_{\mathcal{A}_1,\mathcal{S}_2}
& u_{10}\bigoplus v_{8}= X^{(0,2)}_{\mathcal{S}_1,\mathcal{B}_2}\bigoplus Y^{(2,0)}_{\mathcal{A}_2,\mathcal{S}_2}
& u_{11}\bigoplus v_{11}=X^{(0,3)}_{\mathcal{S}_1,\mathcal{B}_2}\bigoplus Y^{(2,0)}_{\mathcal{A}_3,\mathcal{S}_2}\\[0.2cm]
  u_{12}\bigoplus v_{12}=X^{(1,0)}_{\mathcal{S}_1,\mathcal{B}_0}\bigoplus Y^{(0,1)}_{\mathcal{A}_0,\mathcal{S}_2}
& u_{13}\bigoplus v_{15}=X^{(1,1)}_{\mathcal{S}_1,\mathcal{B}_0}\bigoplus Y^{(0,1)}_{\mathcal{A}_1,\mathcal{S}_2}
& u_{14}\bigoplus v_{18}=X^{(1,2)}_{\mathcal{S}_1,\mathcal{B}_0}\bigoplus Y^{(0,1)}_{\mathcal{A}_2,\mathcal{S}_2}\\[0.2cm]
  u_{15}\bigoplus v_{21}=X^{(1,3)}_{\mathcal{S}_1,\mathcal{B}_0}\bigoplus Y^{(0,1)}_{\mathcal{A}_3,\mathcal{S}_2}
& u_{16}\bigoplus v_{13}=X^{(1,0)}_{\mathcal{S}_1,\mathcal{B}_1}\bigoplus Y^{(1,1)}_{\mathcal{A}_0,\mathcal{S}_2}
& u_{17}\bigoplus v_{16}=X^{(1,1)}_{\mathcal{S}_1,\mathcal{B}_1}\bigoplus Y^{(1,1)}_{\mathcal{A}_1,\mathcal{S}_2}\\[0.2cm]
  u_{18}\bigoplus v_{19}=X^{(1,2)}_{\mathcal{S}_1,\mathcal{B}_1}\bigoplus Y^{(1,1)}_{\mathcal{A}_2,\mathcal{S}_2}
& u_{19}\bigoplus v_{22}=X^{(1,3)}_{\mathcal{S}_1,\mathcal{B}_1}\bigoplus Y^{(1,1)}_{\mathcal{A}_3,\mathcal{S}_2}
& u_{20}\bigoplus v_{14}=X^{(1,0)}_{\mathcal{S}_1,\mathcal{B}_2}\bigoplus Y^{(2,1)}_{\mathcal{A}_0,\mathcal{S}_2}\\[0.2cm]
  u_{21}\bigoplus v_{17}=X^{(1,1)}_{\mathcal{S}_1,\mathcal{B}_2}\bigoplus Y^{(2,1)}_{\mathcal{A}_1,\mathcal{S}_2}
& u_{22}\bigoplus v_{20}=X^{(1,2)}_{\mathcal{S}_1,\mathcal{B}_2}\bigoplus Y^{(2,1)}_{\mathcal{A}_2,\mathcal{S}_2}
& u_{23}\bigoplus v_{23}=X^{(1,3)}_{\mathcal{S}_1,\mathcal{B}_2}\bigoplus Y^{(2,1)}_{\mathcal{A}_3,\mathcal{S}_2}
\end{array}
\end{eqnarray*}
\end{small}Finally for the remaining vertices in $\mathcal{Y}$ of $\mathbf{G}=(\mathcal{X},\mathcal{Y};\mathcal{E})$, the server sends the following $8$ coded messages,
\begin{small}
\begin{eqnarray*}
\begin{array}{cccccc}
v_{24}\bigoplus v_{25}=Y^{(0,2)}_{\mathcal{A}_0,\mathcal{S}_2}\bigoplus Y^{(1,2)}_{\mathcal{A}_0,\mathcal{S}_2}&
v_{25}\bigoplus v_{26}=Y^{(1,2)}_{\mathcal{A}_0,\mathcal{S}_2}\bigoplus Y^{(2,2)}_{\mathcal{A}_0,\mathcal{S}_2}\\[0.2cm]
v_{27}\bigoplus v_{28}=Y^{(0,2)}_{\mathcal{A}_1,\mathcal{S}_2}\bigoplus Y^{(1,2)}_{\mathcal{A}_1,\mathcal{S}_2}&
v_{28}\bigoplus v_{29}=Y^{(1,2)}_{\mathcal{A}_1,\mathcal{S}_2}\bigoplus Y^{(2,2)}_{\mathcal{A}_1,\mathcal{S}_2}\\[0.2cm]
v_{30}\bigoplus v_{31}=Y^{(0,2)}_{\mathcal{A}_2,\mathcal{S}_2}\bigoplus Y^{(1,2)}_{\mathcal{A}_2,\mathcal{S}_2}&
v_{31}\bigoplus v_{32}=Y^{(1,2)}_{\mathcal{A}_2,\mathcal{S}_2}\bigoplus Y^{(2,2)}_{\mathcal{A}_2,\mathcal{S}_2}\\[0.2cm]
v_{33}\bigoplus v_{34}=Y^{(0,2)}_{\mathcal{A}_3,\mathcal{S}_2}\bigoplus Y^{(1,2)}_{\mathcal{A}_3,\mathcal{S}_2}&
v_{34}\bigoplus v_{35}=Y^{(1,2)}_{\mathcal{A}_3,\mathcal{S}_2}\bigoplus Y^{(2,2)}_{\mathcal{A}_3,\mathcal{S}_2}
\end{array}
\end{eqnarray*}
\end{small}
\end{itemize}
After receiving all the coded messages, and together with its cache content $\mathcal{Z}_k$, each user is able to reconstruct requested file. Let us take the user $0$ as an example. By $\mathcal{Z}_{0}$, user $0$ only requires the following $18$ sub-packets
$$W^{(l_2,l_1)}_{0,\mathcal{A}_0,\mathcal{B}_0}, \ \ W^{(l_2,l_1)}_{0,\mathcal{A}_0,\mathcal{B}_1},\ \ W^{(l_2,l_1)}_{0,\mathcal{A}_0,\mathcal{B}_2},\ \ \ \ \ \  l_1\in [0,3),\ \ l_2\in[0,2).$$
We can check that the above sub-packets can be decoded from the following $18$ coded messages.
\begin{small}
\begin{eqnarray*}
\begin{array}{cccccc}
u_1\bigoplus v_3 & u_2\bigoplus v_6 & u_3\bigoplus v_9 & u_5\bigoplus v_4 & u_6\bigoplus v_7 & u_7\bigoplus v_{10}\\[0.2cm]
u_9\bigoplus v_5 & u_{10}\bigoplus v_8 & u_{11}\bigoplus v_{11} & u_{13}\bigoplus v_{15} & u_{14}\bigoplus v_{18} & u_{15}\bigoplus v_{21}\\[0.2cm]
u_{17}\bigoplus v_{16} & u_{18}\bigoplus v_{19} & u_{19}\bigoplus v_{22} & u_{21}\bigoplus v_{17} & u_{22}\bigoplus v_{20} & u_{23}\bigoplus v_{23}
\end{array}
\end{eqnarray*}
\end{small}
Furthermore, for each of the above coded message, user $0$ can decode a unique required sub-packet. For instance by
\begin{eqnarray*}\begin{small}
\begin{array}{c}
u_1\bigoplus v_3=X^{(0,1)}_{\mathcal{S}_1,\mathcal{B}_0}\bigoplus Y^{(0,0)}_{\mathcal{A}_1,\mathcal{S}_2}
=\left(W^{(0,0)}_{0,\mathcal{A}_0,\mathcal{B}_0}\bigoplus W^{(0,1)}_{2,\mathcal{A}_2,\mathcal{B}_0}\bigoplus W^{(0,1)}_{3,\mathcal{A}_3,\mathcal{B}_0}\right)\bigoplus \left(W^{(0,0)}_{5,\mathcal{A}_1,\mathcal{B}_2}\bigoplus W^{(0,0)}_{6,\mathcal{A}_1,\mathcal{B}_3}\right)
\end{array}\end{small}
\end{eqnarray*}
clearly user $0$ can decode $W^{(0,0)}_{0,\mathcal{A}_0,\mathcal{B}_0}$, since it has cached all the other sub-packets in $u_1\bigoplus v_3$. Thus finally, user $0$ can decode all the required sub-packets.

Similarly, let us take the  user $4$ as another example. After placement, user $4$ only needs the following $24$ sub-packets
$$W^{(l_2,l_1)}_{4,\mathcal{A}_0,\mathcal{B}_0}, \ \ W^{(l_2,l_1)}_{4,\mathcal{A}_1,\mathcal{B}_0},\ \ W^{(l_2,l_1)}_{4,\mathcal{A}_2,\mathcal{B}_0},\ \ W^{(l_2,l_1)}_{4,\mathcal{A}_3,\mathcal{B}_0}\ \ \ \ \ \  l_1\in [0,3),\ \ l_2\in[0,2),$$
which can be found in the following $24$ coded messages respectively.
\begin{small}
\begin{eqnarray}
\label{exm-2-k2-1}
\begin{array}{cccccccc}
u_4\bigoplus v_1 & u_5\bigoplus v_4 & u_6\bigoplus v_7 & u_7\bigoplus v_{10} & u_8\bigoplus v_2 & u_9\bigoplus v_5 &
u_{10}\bigoplus v_8 & u_{11}\bigoplus v_{11} \\[0.2cm]
u_{16}\bigoplus v_{13} & u_{17}\bigoplus v_{16} & u_{18}\bigoplus v_{19} & u_{19}\bigoplus v_{22}
&u_{20}\bigoplus v_{14} & u_{21}\bigoplus v_{17} & u_{22}\bigoplus v_{20} & u_{23}\bigoplus v_{23}
\end{array}
\end{eqnarray}
\end{small}
\begin{small}
\begin{eqnarray}
\label{exm-2-k2-2}
\begin{array}{cccccccc}
v_{24}\bigoplus v_{25} & v_{25}\bigoplus v_{26} & v_{27}\bigoplus v_{28} & v_{28}\bigoplus v_{29} & v_{30}\bigoplus v_{31} & v_{31}\bigoplus v_{32} & v_{33}\bigoplus v_{34} & v_{34}\bigoplus v_{35}
\end{array}
\end{eqnarray}
\end{small} On the one hand, user $4$ can obtain the following sub-packets by \eqref{exm-2-k2-1}.
$$W^{(l_2,l_1)}_{4,\mathcal{A}_0,\mathcal{B}_0}, \ \ W^{(l_2,l_1)}_{4,\mathcal{A}_1,\mathcal{B}_0},\ \ W^{(l_2,l_1)}_{4,\mathcal{A}_2,\mathcal{B}_0},\ \ W^{(l_2,l_1)}_{4,\mathcal{A}_3,\mathcal{B}_0}\ \ \ \ \ \  l_1\in [0,2),\ \ l_2\in[0,2).$$
and all the remaining
$$W^{(l_2,2)}_{4,\mathcal{A}_0,\mathcal{B}_0},\  W^{(l_2,2)}_{4,\mathcal{A}_1,\mathcal{B}_0},\  W^{(l_2,2)}_{4,\mathcal{A}_2,\mathcal{B}_0},\  W^{(l_2,2)}_{4,\mathcal{A}_3,\mathcal{B}_0}\ \ \ l_2\in[0,2)$$
can be obtained by \eqref{exm-2-k2-2} on the other hand. For instance, by
\begin{small}
\begin{eqnarray}
\label{eq-ex-neq}
\begin{split}
v_{24}\bigoplus v_{25}=Y^{(0,2)}_{\mathcal{A}_0,\mathcal{S}_2}\bigoplus Y^{(1,2)}_{\mathcal{A}_0,\mathcal{S}_2}=\left(W^{(0,2)}_{5,\mathcal{A}_0,\mathcal{B}_1}\bigoplus W^{(0,2)}_{6,\mathcal{A}_0,\mathcal{B}_2}\right)\bigoplus\left(W^{(0,2)}_{4,\mathcal{A}_0,\mathcal{B}_0}\bigoplus W^{(1,2)}_{6,\mathcal{A}_0,\mathcal{B}_2}\right)\\
v_{25}\bigoplus v_{26}=Y^{(1,2)}_{\mathcal{A}_0,\mathcal{S}_2}\bigoplus Y^{(2,2)}_{\mathcal{A}_0,\mathcal{S}_2}=\left(W^{(0,2)}_{4,\mathcal{A}_0,\mathcal{B}_0}\bigoplus W^{(1,2)}_{6,\mathcal{A}_0,\mathcal{B}_2}\right)\bigoplus \left(W^{(1,2)}_{4,\mathcal{A}_0,\mathcal{B}_0}\bigoplus W^{(1,2)}_{5,\mathcal{A}_0,\mathcal{B}_1}\right)
\end{split}\end{eqnarray}
\end{small}
user $4$ can firstly decode $W^{(0,2)}_{4,\mathcal{A}_0,\mathcal{B}_0}$ from $v_{24}\bigoplus v_{25}$. Based on the decoded $W^{(0,2)}_{4,\mathcal{A}_0,\mathcal{B}_0}$ and $v_{25}\bigoplus v_{26}$, user $4$ can further decode $W^{(1,2)}_{4,\mathcal{A}_0,\mathcal{B}_0}$. With the similar strategy, other users can also decode their requested sub-packets.

\end{example}

In this example, by using the proposed scheme, the transmission rate is $\frac{24+8}{72}=\frac{4}{9}$, which is smaller than $\frac{7}{12}$ in Example \ref{ex-trivial}. And this implies that, by using our proposed CSM scheme, the coded messages are designed to better exploit the coded multicasting opportunities among all the users in $\mathcal{K}_1,\mathcal{K}_2$.

\section{Performance Analysis}\label{sec-first-performance}
In this section, we will show that the rate of our new scheme is order optimal. Firstly,  let us introduce the scheme in \cite{SWXH}. In order to simplify, we denote the cache size of fixed  user and  group user as $M_{\mathcal{K}_1}$, $M_{\mathcal{K}_2}$ in this section.
\begin{lemma}\rm(Order-optimal scheme in \cite{SWXH})
\label{le-Decentral}
For a caching problem with $N$ files and $K$ users with $(K<N)$, without loss of generality, assume that the caching sizes satisfy $M_{0}\leq M_{1}\leq \ldots \leq M_{K-1}$ and denote $\mathcal{M}=\{M_{0}, M_{1}, \ldots, M_{K-1}\}$. A scheme can be obtained with transmission rate
\begin{eqnarray}
\label{eq-decentral-R}
R_D(K;\mathcal{M};N)=\sum\limits_{i=0}^{K-1}\left[  \prod\limits_{j=0}^{i}\left(1-\frac{M_j}{N}\right)\right].
\end{eqnarray} Furthermore this scheme is order optimal since
\begin{eqnarray}
\label{eq-order-decentral}
\frac{R_D(K;\mathcal{M};N)}{R^*(K;\mathcal{M};N)}\leq 6.
\end{eqnarray}
\end{lemma}

 Assume that $\mathcal{M}=\{M_{\mathcal{K}_1},\ldots,M_{\mathcal{K}_1},M_{\mathcal{K}_2},\ldots,M_{\mathcal{K}_2}\}$. Let $\lambda_1=\frac{M_{\mathcal{K}_1}}{N}$, $\lambda_2=\frac{M_{\mathcal{K}_2}}{N}$ and $K=K_1+K_2$. Let us show the order optimality and smaller rate compared with the scheme from \cite{SWXH} in the following separate cases.

\subsection{The case $\lambda_1\leq\lambda_2$}\label{subsec-M1<M2}
When $\lambda_1\leq\lambda_2$, the cache set can be denoted by $$\mathcal{M}=
\{\underbrace{M_{\mathcal{K}_1},\ldots,M_{\mathcal{K}_1}}_{K_1},\underbrace{M_{\mathcal{K}_2},\ldots,M_{\mathcal{K}_2}}_{K_2}\},$$
then \eqref{eq-decentral-R} can be written as
\begin{small}
\begin{eqnarray}
\begin{split}
\label{eq-decentral-R-K12-1}
R_D(K;\mathcal{M};N)
=&\frac{K_1-t_1}{t_1}\left[1-\left(1-\frac{t_1}{K_1}\right)^{K_1}\right]+\left(1-\frac{t_1}{K_1}\right)^{K_1}\frac{K_2-t_2}{t_2}\left[1-\left(1-\frac{t_2}{K_2}\right)^{K_2}\right]\\
=&\left(\frac{1}{\lambda_1}-1\right)\left[1-\left(1-\lambda_1\right)^{K_1}\right]+\left(1-\lambda_1\right)^{K_1}\left(\frac{1}{\lambda_2}-1\right)\left[1-\left(1-\lambda_2\right)^{K_2}\right].
\end{split}\end{eqnarray}
\end{small} That is, from Lemma \ref{le-Decentral} we have an order optimal scheme with rate in \eqref{eq-decentral-R-K12-1}, and by Theorem \ref{th-new-r} we have
\begin{small}
\begin{eqnarray}
\label{eq-th-new-Simplify-1}
\begin{split}
R(K;\mathcal{M};N)&=\frac{(K_1-t_1)(t_1+t_2+1)}{(t_1+1)t_2t_1}+\frac{(K_2-t_2)t_1}{t_2(t_1+1)}\\
&=\left(\frac{1}{\lambda_1}-1\right)\frac{K_1\lambda_1+K_2\lambda_2+1}{(K_1\lambda_1+1)K_2\lambda_2}+\left(\frac{1}{\lambda_2}-1\right)\frac{K_1\lambda_1}{(K_1\lambda_1+1)}\\
&=\left(\frac{1}{\lambda_1}-1\right)\left(\frac{1}{K_2\lambda_2}+\frac{1}{K_1\lambda_1+1}\right)+\left(\frac{1}{\lambda_2}-1\right)\left(1-\frac{1}{K_1\lambda_1+1}\right)
\end{split}
\end{eqnarray}
\end{small}
By \eqref{eq-decentral-R-K12-1} and \eqref{eq-th-new-Simplify-1}, we have
\begin{small}
\begin{eqnarray}
\begin{split}
\label{eq-minus-decentral-new-1}
&\frac{R(K;\mathcal{M};N)}{R_D(K;\mathcal{M};N)}\\
=&\frac{\left(\frac{1}{\lambda_1}-1\right)\left(1-\frac{1}{K_1\lambda_1+1}+\frac{1}{K_2\lambda_2}\right)+\left(\frac{1}{\lambda_2}-1\right)\frac{1}{K_1\lambda_1+1}}
{\left(\frac{1}{\lambda_1}-1\right)\left[1-(1-\lambda_1)^{K_1}\right]+(1-\lambda_1)^{K_1}\left(\frac{1}{\lambda_2}-1\right)\left[1-(1-\lambda_2)^{K_2}\right]}\\
<&\frac{\left(\frac{1}{\lambda_1}-1\right)\left(1-\frac{1}{K_1\lambda_1+1}+\frac{1}{K_2\lambda_2}+\frac{1}{K_1\lambda_1+1}\right)}
{\left(\frac{1}{\lambda_1}-1\right)\left[1-(1-\lambda_1)^{K_1}\right]}\\
<&\frac{\left(\frac{1}{\lambda_1}-1\right)\left(1+\frac{1}{K_2\lambda_2}\right)}{\left(\frac{1}{\lambda_1}-1\right)\left(1-\frac{1}{K_1\lambda_1}\right)}\\
<&\frac{1+\frac{1}{K_2\lambda_2}}{1-\frac{1}{K_1\lambda_1}}\\
<&2
\end{split}\end{eqnarray}\end{small}

Now we consider the parameter where the rate of our new scheme is smaller than the rate of the scheme from \cite{SWXH}, i.e., $\frac{R(K;\mathcal{M};N)}{R_D(K;\mathcal{M};N)}<1$. Similarly, according to \eqref{eq-decentral-R-K12-1} and \eqref{eq-th-new-Simplify-1}, we have
\begin{eqnarray*}
\begin{split}
&\frac{R(K;\mathcal{M};N)}{R_D(K;\mathcal{M};N)}\\[0.2cm]
=&\frac{\left(\frac{1}{\lambda_1}-1\right)\left(1-\frac{1}{K_1\lambda_1+1}+\frac{1}{K_2\lambda_2}\right)+
\left(\frac{1}{\lambda_2}-1\right)\frac{1}{K_1\lambda_1+1}}
{\left(\frac{1}{\lambda_1}-1\right)\left[1-\left(1-\lambda_1\right)^{K_1}\right]+
\left(1-\lambda_1\right)^{K_1}\left(\frac{1}{\lambda_2}-1\right)\left[1-\left(1-\lambda_2\right)^{K_2}\right]}\\[0.2cm]
=&\frac{\left(\frac{1}{\lambda_1}-1\right)\left[-\frac{1}{K_1\lambda_1+1}+\frac{1}{K_2\lambda_2}+\left(1-\lambda_1\right)^{K_1}\right]+
\left(\frac{1}{\lambda_2}-1\right)\left[\frac{1}{K_1\lambda_1+1}-\left(1-\lambda_1\right)^{K_1}+\left(1-\lambda_1\right)^{K_1}\left(1-\lambda_2\right)^{K_2}\right]}
{\left(\frac{1}{\lambda_1}-1\right)\left[1-\left(1-\lambda_1\right)^{K_1}\right]+
\left(1-\lambda_1\right)^{K_1}\left(\frac{1}{\lambda_2}-1\right)\left[1-\left(1-\lambda_2\right)^{K_2}\right]}+1\\[0.2cm]
=&\frac{\left(\frac{1}{\lambda_1}-1\right)x+\left(\frac{1}{\lambda_2}-1\right)y}{z}+1
\end{split}\end{eqnarray*}
where $x=-\frac{1}{K_1\lambda_1+1}+\frac{1}{K_2\lambda_2}+\left(1-\lambda_1\right)^{K_1}$, $y=\frac{1}{K_1\lambda_1+1}-\left(1-\lambda_1\right)^{K_1}+\left(1-\lambda_1\right)^{K_1}\left(1-\lambda_2\right)^{K_2}$ and $z=1\left/\right.\left(\frac{1}{\lambda_1}-1\right)\left[1-\left(1-\lambda_1\right)^{K_1}\right]+
\left(1-\lambda_1\right)^{K_1}\left(\frac{1}{\lambda_2}-1\right)\left[1-\left(1-\lambda_2\right)^{K_2}\right]$. It is not difficult to check that $y>0$ and $z>0$ always hold. By calculating we can see that $\frac{1}{K_2\lambda_2+1}>\left(1-\lambda_2\right)^{K_2}$ holds when $K_2$ is enough large. We claim that if $x<0$ and
\begin{eqnarray}
\label{eq-com-neq4}
K_1(1-\lambda_1)+K_2(1-\lambda_2)+\frac{1}{\lambda_1}+\frac{1}{\lambda_2}<2+\frac{K_1}{\lambda_1},\end{eqnarray}then
\begin{eqnarray}\label{eq-com-neq1}
\left(\frac{1}{\lambda_1}-1\right)x+\left(\frac{1}{\lambda_2}-1\right)y<0.
\end{eqnarray}
This implies that $\frac{R(K;\mathcal{M};N)}{R_D(K;\mathcal{M};N)}<1$. Since
\begin{eqnarray}
\label{eq-com-neq3}
\begin{split}
\frac{-x}{y}
=&\frac{\frac{1}{K_1\lambda_1-1}+\frac{1}{K_2\lambda_2}-\left(1-\lambda_1\right)^{K_1}}
{\frac{1}{K_1\lambda_1+1}-\left(1-\lambda_1\right)^{K_1}+\left(1-\lambda_1\right)^{K_1}\left(1-\lambda_2\right)^{K_2}}\\[0.2cm]
\geq&\frac{\frac{1}{K_1\lambda_1-1}+\frac{1}{K_2\lambda_2}}
{\frac{1}{K_1\lambda_1+1}+\left(1-\lambda_1\right)^{K_1}\left(1-\lambda_2\right)^{K_2}}\\[0.2cm]
\geq&\frac{\frac{1}{K_1\lambda_1-1}+\frac{1}{K_2\lambda_2}}
{\frac{1}{K_1\lambda_1+1}+\frac{1}{K_1\lambda_1+1}\frac{1}{K_2\lambda_2+1}}\\[0.2cm]
=&\frac{K_2\lambda_2-K_1\lambda_1-1}{K_2\lambda_2+1},
\end{split}
\end{eqnarray} and \eqref{eq-com-neq4} can be written as
$$\frac{\frac{1}{\lambda_2}-1}{\frac{1}{\lambda_1}-1}<\frac{K_2\lambda_2-K_1\lambda_1-1}{K_2\lambda_2+1},$$
we have \begin{eqnarray}
\label{eq-com-neq2}
\frac{\frac{1}{\lambda_2}-1}{\frac{1}{\lambda_1}-1}<\frac{-x}{y},
\end{eqnarray} i.e., \eqref{eq-com-neq1} always holds.
\subsection{The case $\lambda_1>\lambda_2$}\label{subsec-M1>M2}
When $\lambda_1>\lambda_2$, the coresponding cache set can also be denoted by $$\mathcal{M}'=
\{\underbrace{M_{\mathcal{K}_2},\ldots,M_{\mathcal{K}_2}}_{K_2},\underbrace{M_{\mathcal{K}_1},\ldots,M_{\mathcal{K}_1}}_{K_1}\},$$
Similar to \eqref{eq-decentral-R-K12-1}, \eqref{eq-decentral-R} can be rewritten as
\begin{small}
\begin{eqnarray*}
\begin{split}
R_D(K;\mathcal{M}';N)=&\left(\frac{1}{\lambda_2}-1\right)\left[1-\left(1-\lambda_2\right)^{K_2}\right]+\left(1-\lambda_2\right)^{K_2}\left(\frac{1}{\lambda_1}-1\right)\left[1-\left(1-\lambda_1\right)^{K_1}\right].
\end{split}\end{eqnarray*}
\end{small}
According to Theorem \ref{th-new-r} we have
\begin{small}
\begin{eqnarray*}
\begin{split}
R(K;\mathcal{M}';N)&=\frac{(K_1-t_1)t_2}{t_1(t_2+1)}+\frac{(K_2-t_2)(t_1+t_2+1)}{(t_2+1)t_2t_1}\\
&=\left(\frac{1}{\lambda_1}-1\right)\frac{K_2\lambda_2}{(K_2\lambda_2+1)}+\left(\frac{1}{\lambda_2}-1\right)\frac{K_1\lambda_1+K_2\lambda_2+1}{(K_2\lambda_2+1)K_1\lambda_1}\\
&=\left(\frac{1}{\lambda_1}-1\right)\left(1-\frac{1}{K_2\lambda_2+1}\right)+\left(\frac{1}{\lambda_2}-1\right)\left(\frac{1}{K_1\lambda_1}+\frac{1}{K_2\lambda_2+1}\right)
\end{split}
\end{eqnarray*}
\end{small}
Similar to \eqref{eq-minus-decentral-new-1}, we have
\begin{small}
\begin{eqnarray}
\label{eq-minus-decentral-new-2}
\begin{split}
\frac{R(K;\mathcal{M}';N)}{R_D(K;\mathcal{M}';N)}<2\\
\end{split}\end{eqnarray}\end{small} Then using the same analysis in the Subsection \ref{subsec-M1<M2}, we claim that $\frac{R(K;\mathcal{M}';N)}{R_D(K;\mathcal{M}';N)}<1$ always holds when $\lambda_1$ and $\lambda_2$ satisfy the following conditions.
\begin{eqnarray}
\label{eq-3}
K_1(1-\lambda_1)+K_2(1-\lambda_2)+\frac{1}{\lambda_1}+\frac{1}{\lambda_2}<2+\frac{K_2}{\lambda_2}
\end{eqnarray}
\subsection{Order-optimality and Smaller Rate}\label{sub-sec-Order-Smaller}
From the discussions of Subsections \ref{subsec-M1<M2} and \ref{subsec-M1>M2}, the statements holds.
\begin{remark}\rm
\label{rem:1}
\begin{eqnarray*}
\frac{R(K;\mathcal{M};N)}{R^*(K;\mathcal{M};N)}<\frac{2R_D(K;\mathcal{M};N)}{R^*(K;\mathcal{M};N)}< 12.
\end{eqnarray*}
Therefore, the scheme in Theorem \ref{th-new-r} is order optimal.
\end{remark}

\begin{remark}\label{rem:2}
If $\lambda_1$ and $\lambda_2$ satisfy the following conditions,
\begin{align}\label{eq:rem2}
K_1(1-\lambda_1)+K_2(1-\lambda_2)+\frac{1}{\lambda_1}+\frac{1}{\lambda_2}<\left\{
             \begin{array}{lr}
             2+\frac{K_1}{\lambda_1},  &\lambda_1 \leq\lambda_2\\
             2+\frac{K_2}{\lambda_2},  &\lambda_1 >\lambda_2
             \end{array}
\right.,
\end{align}
the scheme in Theorem \ref{th-new-r} has the smaller transmission rate than that of the scheme proposed in Lemma \ref{le-Decentral}, by \eqref{eq-com-neq4} and \eqref{eq-3}.
\end{remark}
In fact, there is a large region of $\lambda_1$ and $\lambda_2$ satisfying \eqref{eq:rem2}. For example, let $K_1=500$ and $K_2=100$. Assume that $1>\lambda_1>\lambda_2$. The range of $\lambda_1$ and $\lambda_2$ satisfying \eqref{eq:rem2} is shown in Figure \ref{fig-7}.
\begin{figure}[htbp]
  \centering
  \includegraphics[width=0.5\textwidth]{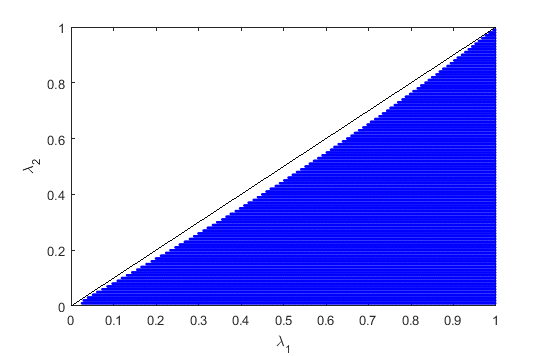}
  \caption{The range of $\lambda_1$ and $\lambda_2$ satisfying \eqref{eq:rem2}.}\label{fig-7}
\end{figure}
\section{Conclusion}\label{sec-conclusion}
In this paper, we focus on the coded caching design for the system containing a set of fixed and mobile users. In order to avoid resource consumption resulted from unnecessary cache adaptation and updating for those fixed users in the placement phase, and to minimize the amount of transmission in the delivery phase, we propose a $(K_1,K_2;M_1,M_2;N)$ dynamic coded caching design framework through concatenating-based placement and the saturating matching based delivery. With the proposed scheme, once the cache of fixed users is fulfilled, there is no need to consider the cache content adaptation when there is no change in the files at server. Instead, the placement phase only needs to be performed for those mobile users. On the basis, the saturating matching based multicasting can be performed to minimize the transmission load in the delivery phase. We have also shown that the proposed scheme is order-optimal.

\section*{Appendix: The proof of Theorem \ref{th-new-r}}
\label{sec-appendix}
In order to prove Theorem \ref{th-new-r}, the following results are useful. First let us consider each edge defined in Algorithm \ref{alg-NSD}.
\begin{proposition}
\label{le-proof-pair}
For any $\mathcal{S}_1=\{k_{1,0},k_{1,1},\ldots,k_{1,t_1}\}\in \mathfrak{S}_1$ and $\mathcal{S}_2=\{k_{2,0},k_{2,1},\ldots,k_{2,t_2}\}\in \mathfrak{S}_2$, by \eqref{eq-ex-step1-0} and \eqref{eq-ex-step1-1}, we have an edge $(X^{(l_2,m_1)}_{\mathcal{S}_1,\mathcal{B}},Y^{(m_2,l_1)}_{\mathcal{A},\mathcal{S}_2})$. From the Algorithm
\ref{alg-NSD}, the server transmits
\begin{eqnarray}
\begin{split}
\label{eq-proof-signal}
X^{(l_2,m_1)}_{\mathcal{S}_1,\mathcal{B}}\bigoplus Y^{(m_2,l_1)}_{\mathcal{A},\mathcal{S}_2}
&=\left(\bigoplus\limits_{m'_1=0}^{m_1-1} W^{(l_2,m_1-1)}_{d_{k_{1,m_1'}},\mathcal{S}_1\setminus\{k_{1,m'_1}\},\mathcal{B}}\right)\bigoplus\left(
\bigoplus\limits_{m'_1=m_1+1}^{t_1} W^{(l_2,m_1)}_{d_{k_{1,m_1'}},\mathcal{S}_1\setminus\{k_{1,m'_1}\},\mathcal{B}}\right)\\
&\bigoplus\left(\bigoplus\limits_{m'_2=0}^{m_2-1} W^{(m_2-1,l_1)}_{d_{k_{2,m'_2}},\mathcal{A},\mathcal{S}_2\setminus\{k_{2,m'_2}\}}\right)\bigoplus\left(
\bigoplus\limits_{m'_2=m_2+1}^{t_2} W^{(m_2,l_1)}_{d_{k_{2,m'_2}},\mathcal{A},\mathcal{S}_2\setminus\{k_{2,m'_2}\}}\right)
\end{split}
\end{eqnarray}
where $\mathcal{A}=\mathcal{S}_{1}\setminus\{k_{1,m_1}\}$ and $\mathcal{B}=\mathcal{S}_{2}\setminus\{k_{2,m_2}\}$. Then in delivery phase, each user in $\mathcal{A}\bigcup\mathcal{B}$ can obtain its requested sub-packet from the received signal in \eqref{eq-proof-signal}.
\end{proposition}
\begin{proof} Let us consider any user $k_{1,s'_1}\in \mathcal{A}$, $0\leq s'_1\leq t_1$. By \eqref{eq-proof-signal}, if $s'_1<m_1$, the sub-packet $W^{(l_2,m_1-1)}_{d_{k_{1,s'_1}},\mathcal{S}_1\setminus k_{1,s'_1},\mathcal{B}}$ is required by user $k_{1,s'_1}$. Otherwise $W^{(l_2,m_1)}_{d_{k_{1,s'_1}},\mathcal{S}_1\setminus k_{1,s'_1},\mathcal{B}}$ is required. From Line $6$ in Algorithm \ref{alg-NSP}, user $k_{1,s'_1}$ caches other sub-packets in  \eqref{eq-proof-signal}, \emph{i.e.}, $W^{(l_2,i)}_{d_{k_{1,m'_1}},\mathcal{S}_1\setminus \{k_{1,m'_1}\},\mathcal{B}}$ and $W^{(j,l_1)}_{d_{k_{2,m'_2}},\mathcal{A},\mathcal{S}_2\setminus \{k_{2,m'_2}\}}$ where $s'_1\neq m'_1$, $i=m_1-1$, $m_1$ and  $j=m_2-1$, $m_2$. So the required sub-packet can be decoded by user $k_{1,s'_1}$. Similarly each user from $\mathcal{B}$ can also decode its required sub-packet.
\end{proof}

\begin{lemma}
\label{le-hall-theorem}
(\cite{BM})Given a bipartite graph $\mathcal{G}=(\mathcal{X},\mathcal{Y};\mathcal{E})$, if there exist two positive integers $m$ and $n$ such that $d(\mathcal{X})=m$ and $d(\mathcal{Y})=n$, then there is a saturating matching for $\mathcal{X}$ if $m>n$, \emph{i.e.}, there exists a matching with $|\mathcal{X}|$ edges.
\end{lemma}

Without loss of generality, we assume that $M_{\mathcal{K}_1}\geq M_{\mathcal{K}_2}$. Then we can check that $\left\lceil\frac{t_2(K_1-t_1)}{K_2-t_2}\right\rceil<t_1$ always holds. From Line  $14$ in Algorithm \ref{alg-NSD}, we have a bipartite graph $\mathcal{G}=(\mathcal{X},\mathcal{Y};\mathcal{E})$. Now let us consider its induced subgraph $\mathcal{G}'=(\mathcal{X},\mathcal{Y}';\mathcal{E}')$ where
\begin{eqnarray}
\begin{split}
\label{eq-new-bipartiite-prove}
\mathcal{X}&=\left\{X^{(l_2,m_1)}_{\mathcal{S}_{1},\mathcal{B}}\ |\ \mathcal{S}_{1}\subset \mathcal{K}_1,\mathcal{B}\subset\mathcal{K}_2, |\mathcal{S}_1|=t_1+1, |\mathcal{B}|=t_2, l_2\in [0,t_2), m_1\in [0,t_1]\right\}\\
\mathcal{Y}'&=\left\{Y^{(m_2,l_1)}_{\mathcal{A},\mathcal{S}_2 }\ |\ \mathcal{S}_2\subset \mathcal{K}_2, \mathcal{A}\subset \mathcal{K}_1, |\mathcal{S}_2|=t_2+1, |\mathcal{A}|=t_1, l_1\in \left[0,\left\lceil\frac{t_2(K_1-t_1)}{K_2-t_2}\right\rceil\right), m_2\in [0,t_2]\right\}
\end{split}
\end{eqnarray}
\begin{proposition}
\label{pro-baohe-X}
The bipartite graph $\mathcal{G}'=(\mathcal{X},\mathcal{Y}';\mathcal{E}')$ with parameters in \eqref{eq-new-bipartiite-prove} exists a saturating matching for $\mathcal{X}$.
\end{proposition}
\begin{proof}
By Line $14$ in Algorithm \ref{alg-NSD}, each vertex, say $X^{(l_2,m_1)}_{\mathcal{S}_{1},\mathcal{B}}\in \mathcal{X}$, is adjacent to $Y^{(m_2,l_1)}_{\mathcal{S}_{1}\setminus\{k_{1,m_1}\},\mathcal{B}\bigcup \{k_{2,m_2}\} }$ for all $l_1\in \left[0,\left\lceil\frac{t_2(K_1-t_1)}{K_2-t_2}\right\rceil\right)$ and $k_{2,m_2}\in \mathcal{K}_2\setminus \mathcal{B}$. It is easy to check that $d\left(X^{(l_2,m_1)}_{\mathcal{S}_{1},\mathcal{B}}\right)={K_2-t_2 \choose 1}\left\lceil\frac{t_2(K_1-t_1)}{K_2-t_2}\right\rceil$. So we have $d\left(\mathcal{X}\right)={K_2-t_2 \choose 1}\left\lceil\frac{t_2(K_1-t_1)}{K_2-t_2}\right\rceil$. Furthermore each vertex, say $Y^{(m_2,l_1)}_{\mathcal{A},\mathcal{S}_2}$, is adjacent to $X^{(l_2,m_1)}_{\mathcal{A}\bigcup\{k_{1,m_1}\},\mathcal{S}_2\setminus \{k_{2,m_2}\} }$ for all $l_2\in [0,t_2)$ and $k_{1,m_1}\in \mathcal{K}_1\setminus \mathcal{A}$. Similarly we have $d\left(Y^{(m_2,l_1)}_{\mathcal{A},\mathcal{S}_2}\right)={K_1-t_1 \choose 1}t_2$. Then $d(\mathcal{Y}')={K_1-t_1 \choose 1}t_2$. Clearly $$d(\mathcal{X})={K_2-t_2 \choose 1}\left\lceil\frac{t_2(K_1-t_1)}{K_2-t_2}\right\rceil\leq{K_2-t_2 \choose 1}\frac{t_2(K_1-t_1)}{K_2-t_2}=t_1(K_1-t_1)=d(\mathcal{Y}').$$
From Lemma \ref{le-hall-theorem}, there exists a saturating matching for $\mathcal{X}$ in $\mathcal{G}'$.
\end{proof}
Now let us prove Theorem \ref{th-new-r}.
\begin{proof}  The placement phase can be completed by Algorithm \ref{alg-NSP}. For $\mathcal{K}_1$ users, by Lines $1-8$ in Algorithm \ref{alg-NSP}, each users $k_1\in \mathcal{K}_1$ caches a total of ${K_1-1 \choose t_1-1}$ packets for each file, and each file divided into ${K_1 \choose t_1}$ packets. Thus the user $k_1$ caches ${K_1-1 \choose t_1-1}N\left/\right.{K_1 \choose t_1}=\frac{t_1}{K_1}N=M_1$ files. By Lines $13-14$ in Algorithm \ref{alg-NSP}, each users $k_2\in \mathcal{K}_2$ caches a total of $t_1t_2{K_2-1 \choose t_2-1}$ sub-packets for each packet, and each file is divided into $t_1t_2{K_1 \choose t_1}{K_2 \choose t_2}$ sub-packets. Thus the user $k_2$ caches $t_1t_2{K_1 \choose t_1}{K_2-1 \choose t_2-1}N\left/\right.t_1t_2{K_1 \choose t_1}{K_2 \choose t_2}=\frac{t_2}{K_2}N=M_2$ files. Therefore, satisfies they cache size limitation. Then we only consider the delivery phase. Assume $\mathcal{G}''=\{\mathcal{X}, \mathcal{Y}'', \mathcal{E}''\}$ is a saturating matching of $G'=\{\mathcal{X}, \mathcal{Y}', \mathcal{E}'\}$. Using bipartite graphs $\mathcal{G}$ and $\mathcal{G}'$, the delivery phase consists of the following two steps.
\begin{itemize}
\item[Step1:] Let us consider $\mathcal{G}'$ first.
For each edge $(x,y)\in\mathcal{E}''$, $x\in\mathcal{X}$, $y\in \mathcal{Y}''$, the server transmits signal $x\bigoplus y$. Then the left vertices in $\mathcal{Y}'\setminus \mathcal{Y}''$ is transmitted by the server directly.
\item[Step2:] The case $l_1\in\left[0,\left\lceil\frac{t_2(K_1-t_1)}{K_2-t_2}\right\rceil\right)$ has been discussed in Step 1. It is sufficient to consider the case $l_1\in \left[\left\lceil\frac{t_2(K_1-t_1)}{K_2-t_2}\right\rceil,t_1\right)$, \emph{i.e.}, the left vertices in $\mathcal{Y}\setminus \mathcal{Y}'$. By Line $14$ in Algorithm \ref{alg-NSD}, the left vertex set
\begin{eqnarray*}
\mathcal{Y}'''&=\mathcal{Y}\setminus \mathcal{Y}'=\left\{Y^{(m_2,l_1)}_{\mathcal{A},\mathcal{S}_2 }\ |\ \mathcal{S}_2\subset \mathcal{K}_2, \mathcal{A}\subset \mathcal{K}_1, |\mathcal{S}_2|=t_2+1, |\mathcal{A}|=t_1, l_1\in \left[\left\lceil\frac{t_2(K_1-t_1)}{K_2-t_2}\right\rceil,t_1\right), m_2\in [0,t_2]\right\}.
\end{eqnarray*} We divide the vertex set $\mathcal{Y}'''$ into the following order subsets
$$\left(Y^{(0,l_1)}_{\mathcal{A},\mathcal{S}_2},Y^{(1,l_1)}_{\mathcal{A},\mathcal{S}_2},\ldots,
Y^{(t_2,l_1)}_{\mathcal{A},\mathcal{S}_2}\right),\ \ l_1\in\left[\left\lceil\frac{t_2(K_1-t_1)}{K_2-t_2}\right\rceil,t_1\right),\ \ \mathcal{A}\subset\mathcal{K}_1,\ \ |\mathcal{A}|=t_1,\ \ \mathcal{S}_2\subset\mathcal{K}_2,\ \ |\mathcal{S}_2|=t_2+1.$$
Then the server encodes each of the above order sub-set by the following matrix and sends to the users.
\begin{eqnarray*}
\mathbf{G}=\left(
    \begin{array}{ccccc}
      1 &      0      & \cdots & 1 & 1 \\
      1 &      1      & \cdots & 1 & 1 \\
      \vdots & \vdots & \ddots & \vdots & \vdots \\
      1 &      1      & \cdots & 1 & 0 \\
      1 &      1      & \cdots & 1 & 1 \\
      0 &      1      & \cdots & 1 & 1 \\
    \end{array}
  \right)_{(t_2+1)\times t_2}
\end{eqnarray*}For instance, \eqref{eq-ex-neq} by
\begin{eqnarray*}
\begin{split}
\left(Y^{(0,2)}_{\mathcal{A}_0,\mathcal{S}_2}, Y^{(1,2)}_{\mathcal{A}_0,\mathcal{S}_2}, Y^{(2,2)}_{\mathcal{A}_0,\mathcal{S}_2}\right)\cdot\left(
    \begin{array}{cc}
      1 &      0 \\
      1 &      1 \\
      0 &      1 \\
    \end{array}
  \right)
=&\left(Y^{(0,2)}_{\mathcal{A}_0,\mathcal{S}_2}\bigoplus Y^{(1,2)}_{\mathcal{A}_0,\mathcal{S}_2},Y^{(1,2)}_{\mathcal{A}_0,\mathcal{S}_2}\bigoplus Y^{(2,2)}_{\mathcal{A}_0,\mathcal{S}_2} \right)\\
=&\left[\left(W^{(0,2)}_{5,\mathcal{A}_0,\mathcal{B}_1}\bigoplus W^{(0,2)}_{6,\mathcal{A}_0,\mathcal{B}_2}\right)\bigoplus\left(W^{(0,2)}_{4,\mathcal{A}_0,\mathcal{B}_0}\bigoplus W^{(1,2)}_{6,\mathcal{A}_0,\mathcal{B}_2}\right)\right.,\\
&\ \left.\left(W^{(0,2)}_{4,\mathcal{A}_0,\mathcal{B}_0}\bigoplus W^{(1,2)}_{6,\mathcal{A}_0,\mathcal{B}_2}\right)\bigoplus \left(W^{(1,2)}_{4,\mathcal{A}_0,\mathcal{B}_0}\bigoplus W^{(1,2)}_{5,\mathcal{A}_0,\mathcal{B}_1}\right)\right]\\
\end{split}\end{eqnarray*}
We can check that the total number of signal is $S=\left(t_1-\left\lceil\frac{t_2(K_1-t_1)}{K_2-t_2}\right\rceil\right){K_1 \choose t_1}{K_2 \choose t_2+1}\times t_2$.
%
\end{itemize}

Assume that request vector ${\bf d}=(0,1,\ldots,K_1,\ldots, K_1+K_2-1)$. Finally let us consider the request for each user. Without loss of generality, it is sufficient to consider the requests of the user $0$ from $\mathcal{K}_1$ and the user $K_1$ from $\mathcal{K}_2$ respectively based on Step1 and Step2. Let us consider the request of user $0$ first. We claim user $0$ can decode all the sub-packets $$
W^{(l_2,l_1)}_{0,\mathcal{S}_1\setminus\{0\},\mathcal{B}}, \  0\in\mathcal{S}_1\in \mathfrak{S}_1,\ \mathcal{B}\in \mathfrak{B}, \ l_1\in[1,t_1),\ l_2\in [0,t_2),$$ which are not cached, in step1. It is not difficult to check that each $W^{(l_2,l_1)}_{0,\mathcal{S}_1}$ occurs in at most one vertex of $\mathcal{X}$ of graph $\mathcal{G}'$. From Proposition \ref{pro-baohe-X} there exists a the saturating matching for $\mathcal{X}$. So $W^{(l_2,l_1)}_{0,\mathcal{S}_1}$ occurs in at most one edge of the saturating matching. From Proposition \ref{le-proof-pair}, user $0$ can decode each required sub-packet.

Now let us consider user $K_1$'s request. We know that user $K_1$ requests the following sub-packets
$$
W^{(l_2,l_1)}_{K_1,\mathcal{A},\mathcal{S}_2\setminus\{K_1\}}, \  K_1\in\mathcal{S}_2\in \mathfrak{S}_2,\ \mathcal{A}\in \mathfrak{A}, \ l_1\in[1,t_1),\ l_2\in [0,t_2),$$ which are not cached. We also have that each $W^{(l_2,l_1)}_{K_1,\mathcal{A},\mathcal{S}_2\setminus\{K_1\}}$ occurs in at most one vertex of $\mathcal{Y}$ of graph $\mathcal{G}$. We discuss this case by the following two subcases by the definition of induced subgraph $\mathcal{G}'$.
\begin{itemize}
\item $l_1\in \left[0,\left\lceil\frac{t_2(K_1-t_1)}{K_2-t_2}\right\rceil\right)$: Similarly, $W^{(l_2,l_1)}_{K_1,\mathcal{A},\mathcal{S}_2\setminus\{K_1\}}$ occurs in at most one edge of the saturating matching of the induced subgraph $\mathcal{G}'$. And user $K_1$ can decode all the required sub-packets occurring in the edge of the saturating matching. By the proof Proposition  \ref{le-proof-pair}, the left required sub-packets occurring in the vertex of $\mathcal{Y}'$ can also be decoded.

\item $l_1\in \left[\left\lceil\frac{t_2(K_1-t_1)}{K_2-t_2}\right\rceil,t_1\right)$: For each $l_1\in \left[\left\lceil\frac{t_2(K_1-t_1)}{K_2-t_2}\right\rceil,t_1\right)$,  $K_1\in\mathcal{S}_2\in \mathfrak{S}_2$ and $\mathcal{A}\in \mathfrak{A}$, user $K_1$ receives $t_2$ coded signals
\begin{eqnarray}
\label{eq-proof-signal-k2-3}
\left(Y^{(0,l_1)}_{\mathcal{A},\mathcal{S}_2},Y^{(1,l_1)}_{\mathcal{A},\mathcal{S}_2},\ldots,
Y^{(t_2,l_1)}_{\mathcal{A},\mathcal{S}_2}\right)\cdot \mathbf{G}.
\end{eqnarray}
Since $K_1$ is the minimum integer in $\mathcal{S}_2$, the coded signal $Y^{(0,l_1)}_{\mathcal{A},\mathcal{S}_2}$ can be get by user $K_1$, and there exists a unique required sub-packet in each $Y^{(j,l_1)}_{\mathcal{A},\mathcal{S}_2}$, $j\in[1,t_2]$, by \eqref{eq-ex-step1-1}. Since the rank of the matrix $\mathbf{G}$ in \eqref{eq-proof-signal-k2-3} is $t_2$, we can get all the coded signals $Y^{(j,l_1)}_{\mathcal{A},\mathcal{S}_2}$, $j\in[1,t_2]$. By the proof Proposition  \ref{le-proof-pair},  the required sub-packets in $Y^{(j,l_1)}_{\mathcal{A},\mathcal{S}_2}$ can be decoded by user $K_1$.
\end{itemize}

From Setp1 and Step 2, the total times of the transmissions is $$|\mathcal{Y}'|+S=\left\lceil\frac{t_2(K_1-t_1)}{K_2-t_2}\right\rceil(K_2-t_2){K_1 \choose t_1}{K_2 \choose t_2}+\left(t_1-\left\lceil\frac{t_2(K_1-t_1)}{K_2-t_2}\right\rceil\right){K_1 \choose t_1}{K_2 \choose t_2+1}\times t_2$$ when $M_{\mathcal{K}_1}\geq M_{\mathcal{K}_2}$. Then, the rate
\begin{eqnarray}
\label{eq-proof-m1-m2}
\begin{split}
R(K_1,K_2;M_{\mathcal{K}_1},M_{\mathcal{K}_1};N)&=\frac{|\mathcal{Y}'|+S}{t_1t_2{K_1 \choose t_1}{K_2 \choose t_2}}
=\frac{\left\lceil\frac{t_2(K_1-t_1)}{K_2-t_2}\right\rceil(K_2-t_2){K_1 \choose t_1}{K_2 \choose t_2}+(t_1-\left\lceil\frac{t_2(K_1-t_1)}{K_2-t_2}\right\rceil){K_1 \choose t_1}{K_2 \choose t_2+1}t_2}{t_1t_2{K_1 \choose t_1}{K_2 \choose t_2}}\\
&=\frac{\left\lceil\frac{t_2(K_1-t_1)}{K_2-t_2}\right\rceil(K_2-t_2)+(t_1-\left\lceil\frac{t_2(K_1-t_1)}{K_2-t_2}\right\rceil)\frac{K_2-t_2}{t_2+1}t_2}{t_1t_2}\\
&\leq\frac{\left[\frac{t_2(K_1-t_1)}{K_2-t_2}+1\right](K_2-t_2)+\left[t_1-\frac{t_2(K_1-t_1)}{K_2-t_2}\right]\frac{K_2-t_2}{t_2+1}t_2}{t_1t_2}\\
&=\frac{K_1-t_1}{t_1}+\frac{K_2-t_2}{t_1t_2}+\frac{K_2-t_2}{(t_2+1)}-\frac{(K_1-t_1)t_2}{(t_2+1)t_1}\\
&=\frac{(K_1-t_1)}{(t_2+1)t_1}+\frac{(K_2-t_2)(t_1t_2+t_2+1)}{(t_2+1)t_1t_2}
\end{split}
\end{eqnarray}

It is worth noting that when $M_{\mathcal{K}_1}=M_{\mathcal{K}_2}$, we have  $\lceil\frac{t_2(K_1-t_1)}{K_2-t_2}\rceil=t_1$. Submitting this equality into \eqref{eq-proof-m1-m2} we have
\begin{eqnarray}
\label{eq-proof-m1=m2}
R(K_1,K_2;M_{\mathcal{K}_1},M_{\mathcal{K}_1};N)=\frac{K_1-t_1}{t_1}.
\end{eqnarray} Similarly, when $M_{\mathcal{K}_1}< M_{\mathcal{K}_2}$, we have
\begin{eqnarray}
\label{eq-proof-m2-m1}
R(K_1,K_2;M_{\mathcal{K}_1},M_{\mathcal{K}_1};N)\leq\frac{(K_2-t_2)}{(t_1+1)t_2}+\frac{(K_1-t_1)(t_1t_2+t_1+1)}{(t_1+1)t_1t_2}.
\end{eqnarray}
By \eqref{eq-proof-m1-m2} and \eqref{eq-proof-m2-m1}, for any positive integers $K_i$, $N$ and for each $\frac{M_{\mathcal{K}_i}}{N}=\frac{t_i}{K_i}$, $1\leq t_i<K_i$, $i=1,2$, there exists a scheme for $(K_1,K_2;M_{\mathcal{K}_1},M_{\mathcal{K}_2};N)$ caching system
 with transmission rate
\begin{eqnarray*}
R(K_1,K_2;M_{\mathcal{K}_1},M_{\mathcal{K}_2};N)\begin{cases}
\leq\frac{(K_1-t_1)(t_1t_2+t_1+1)}{(t_1+1)t_1t_2}+\frac{(K_2-t_2)}{(t_1+1)t_2},&\text{if $M_{\mathcal{K}_1}<M_{\mathcal{K}_2}$}\\[0.2cm]
\leq\frac{(K_1-t_1)}{(t_2+1)t_1}+\frac{(K_2-t_2)(t_1t_2+t_2+1)}{(t_2+1)t_1t_2},&\text{if $M_{\mathcal{K}_1}\geq M_{\mathcal{K}_2}$}
\end{cases}
\end{eqnarray*}
for a $(K_1,K_2;\mathcal{K}_1,\mathcal{K}_2;N)$ caching system where $t_i=\frac{M_{\mathcal{K}_i} K_i}{N}$.
\end{proof}


\begin{thebibliography}{1}

\bibitem{cisco}
Cisco visual networking index: global mobile data traffic forecast update, 2016-2021 White Paper.

\bibitem{MN}
M. A. Maddah-Ali and U. Niesen, Fundamental limits of caching, {\em IEEE Transactions on Information Theory }, vol. 60, no. 5, pp. 2856šC2867, May. 2014.

\bibitem{GR}
H. Ghasemi and A. Ramamoorthy, Improved lower bounds for coded caching, in Proc. {\em IEEE International Symposium on Information Theory}, Hong Kong, Jun. 2015, pp. 1696-1700.

\bibitem{T}
C. Tian, J. Chen: Caching and delivery via interference elimination. IEEE International Symposium on Information Theory. Barcelona. 830-834, July (2016).

\bibitem{YMA}
Q. Yu, M. A. Maddah-Ali and A. S. Avestimehr, The exact rate-memory tradeoff for caching with uncoded prefetching, {\em IEEE Transactions on Information Theory }, vol. 64, no. 2, pp. 1281-1296, Feb. 2018.

\bibitem{WTP}
K. Wan, D. Tuninetti and P. Piantanida, On the optimality of uncoded cache placement, in Proc. {\em IEEE Information Theory Workshop}, Cambridge, UK, Sept. 2016.

\bibitem{JCLC}
S. Jin, Y. Cui , H. Liu, and G. Caire, Uncoded placement optimization for coded delivery, {\em IEEE WiOpt}, Shanghai, China, May. 2018.

\bibitem{Shariatpanahi_16}
 S. P. Shariatpanahi, S. A. Motahari and B. H. Khalaj, Multi-server coded caching, {\em IEEE Transactions on Information Theory }, vol. 62, no.12, pp. 7253-7271, Dec 2016.

\bibitem{Mital_17}
N. Mital, D. Gunduz and C. Ling, Coded caching in a multi-server system with random topology [online], Available:
https://arxiv.org/abs/1712.00649, Dec 2017.

\bibitem{Ji_16}
Mingyue Ji,  Giuseppe Caire, Andreas F. Molisch, Fundamental Limits of Caching in Wireless D2D Networks, {\em IEEE Transactions on Information Theory }, vol. 62, no. 2, pp. 849-869, 2016.

\bibitem{Karamchandani_16}
N. Karamchandani, U. Niesen, M. A. Maddah-Ali, S. N. Diggavi, Hierarchical Coded Caching, {\em 2014 IEEE International Symposium on Information Theory}, Honolulu, HI, USA, 2014.

\bibitem{Ji_15}
M. Ji, M. Wong, A. M. Tulino, J.Llorca, G. Caire, M. Effros, M. Langberg, On the Fundamental Limits of Caching in Combination Networks, {\em 2015 IEEE 16th International Workshop on Signal Processing Advances in Wireless Communications (SPAWC)}, Stockholm, Sweden, 2015.

\bibitem{Daniel17}
A. M. Daniel, and W. Yu, ``Optimization of heterogeneous coded caching," arXiv:1708.04322.

\bibitem{YCTC}
Q. Yan, M. Cheng, X. Tang and Q. Chen, On the placement delivery array design in centralized coded caching scheme, {\em IEEE Transactions on Information Theory}, vol.~63, no.~9, pp. 5821-5833, 2017.

\bibitem{SZG}
C. Shangguan, Y. Zhang, G. Ge, Centralized coded caching schemes: A hypergraph theoretical approach,  {\em IEEE Transactions on Information Theory,} vol. 64, no. 8, pp. 5755-5766, 2018.

\bibitem{TR}
L. Tang and A. Ramamoorthy, Coded Caching with Low Subpacketization Levels, {\em Globecom Workshops}, Washington, DC, 2016, pp. 1-6.

\bibitem{YTCC}
Q. Yan, X. Tang, Q. Chen, and M. Cheng, Placement delivery array design through strong edge coloring of bipartite graphs, {\em IEEE Communications Letters}, vol. 22, no. 2, pp. 236-239, 2018.

\bibitem{RD}
Reinhard Diestel, Graph Theory Seond Edition, Springer, New York(2000).

\bibitem{SWXH}
S. Wang, W. Li, X. Tian, H. Liu, Fundamental limits of heterogenous cache [online],  Available: http://arxiv.org/abs/1504.01123v1, 2015.



\bibitem{BM}
J. A. Bondy, U. Murthy: Graph Theory with Applications. New York: Elsevier, (1976).

\end{thebibliography}
 \end{document}